%% file: main.tex
\documentclass[sigconf]{acmart}

\usepackage{microtype}

\setcopyright{acmcopyright}
\copyrightyear{2021}
\acmYear{2021}
\setcopyright{acmcopyright}
\acmConference[SIGMOD '21] {Proceedings of the 2021 International Conference on Management of Data}{June 20--25, 2021}{Virtual Event, China}
\acmBooktitle{Proceedings of the 2021 International Conference on Management of Data (SIGMOD '21), June 20--25, 2021, Virtual Event, China}
\acmPrice{15.00}
\acmISBN{978-1-4503-8343-1/21/06}
\acmDOI{10.1145/3448016.3457296}

\usepackage{style}
\input{macro}
\usepackage{bm}

\begin{document}
\fancyhead{}

\title{Fast Parallel Algorithms for Euclidean Minimum Spanning Tree and Hierarchical Spatial Clustering}
\iffullversion
\else
\titlenote{The full version of this paper can be found at \julian{insert arxiv link}.}
\fi

\settopmatter{authorsperrow=4}

\author{Yiqiu Wang}
\affiliation{\institution{MIT CSAIL}}
\email{yiqiuw@mit.edu}
\author{Shangdi Yu}
\affiliation{\institution{MIT CSAIL}}
\email{shangdiy@mit.edu}
\author{Yan Gu}
\affiliation{\institution{UC Riverside}}
\email{ygu@cs.ucr.edu}
\author{Julian Shun}
\affiliation{\institution{MIT CSAIL}}
\email{jshun@mit.edu}

\input{abstract}

\settopmatter{printfolios=true}
\maketitle

\input{intro}
\input{prelim}

\input{hdbscan_wspd}

\input{hdbscan_dendro}

\input{experiment}
\iffullversion
\input{conclusion}
\fi

\myparagraph{Acknowledgements}
We thank Pilar Cano for helpful discussions.
This research was supported by DOE
Early Career Award \#DE-SC0018947, NSF CAREER Award \#CCF-1845763,
Google Faculty Research Award, DARPA SDH Award \#HR0011-18-3-0007, and
Applications Driving Architectures (ADA) Research Center, a JUMP
Center co-sponsored by SRC and DARPA.


\newpage

\bibliographystyle{ACM-Reference-Format}
\bibliography{ref}  

\iffullversion
\newpage
\begin{appendix}

        \input{plane_algorithms}

        \input{subquadratic-appendix}
        \input{approx_optics}

        \input{minpts_ex}

        \input{tables-appendix}
        \input{table_timing_mlpack}
        \input{table_timing_emst}
        \input{table_timing_hdbscan}
\end{appendix}
\fi

\end{document}

%% file: macro.tex
\newcommand{\plot}{reachability plot}
\newcommand{\newddg}{ordered dendrogram}

\newcommand{\wspd}{\textsc{Wspd}\xspace}
\newcommand{\minpts}{\textsf{minPts}\xspace}
\newcommand{\bccpstar}{\textsf{BCCP*}}
\newcommand{\bccp}{\textsf{BCCP}}
\newcommand{\CD}{\textsf{cd}}
\newcommand{\CDmin}{\textsf{cd}_{\min}}
\newcommand{\CDmax}{\textsf{cd}_{\max}}

\newcommand{\kdt}{$k$\text{d-tree}\xspace}

\newcommand{\knn}{$k$\text{-NN}\xspace}

\newcommand{\field}[2]{{#1}_{\textsf{#2}}}

\mathchardef\mhyphen="2D

\algdef{SxnE}{PFor}{EndPFor}[1]{\textbf{par-for} \(\mbox{#1}\) \textbf{do}}
\algdef{SxnE}{InParallel}{EndParallel}{\textbf{do in parallel}}

\newcommand{\codevar}[1]{\mathit{#1}}

\usepackage{soul,color}
\soulregister\cite7
\soulregister\ref7
\soulregister\pageref7
\soulregister\textproc7
\soulregister\textit7
\soulregister\defn7
\soulregister\hdbscan7
\soulregister\myparagraph7
\soulregister\bccp7
\soulregister\newddg7

\newif\iffullversion
\fullversiontrue

\makeatletter
\def\dfnt@space@setup{%
  \dfnt@preskip=\parskip
    \dfnt@postskip=0pt}
\makeatother

\newtheoremstyle{exampstyle}
  {.05in} 
  {.05in} 
  {} 
  {.5em} 
  {\sc \bfseries} 
  {.} 
  {.5em} 
  {} 
\theoremstyle{exampstyle} 
\theoremstyle{exampstyle} 

\makeatletter
\renewenvironment{proof}[1][\proofname]{\par
  \vspace{-\topsep}
  \pushQED{\qed}%
  \normalfont
  \topsep0pt \partopsep0pt 
  \trivlist
  \item[\hskip\labelsep
        \itshape
    #1\@addpunct{.}]\ignorespaces
}{%
  \popQED\endtrivlist\@endpefalse
  \addvspace{3pt plus 3pt} 
}

\crefname{section}{Sec.}{Sec.} 

%% file: abstract.tex
\begin{abstract}

This paper presents new parallel algorithms for generating Euclidean minimum
spanning trees and spatial clustering hierarchies (known as HDBSCAN$^*$). Our approach is based on generating a well-separated pair
decomposition followed by using Kruskal's minimum spanning tree algorithm
and bichromatic closest pair computations. We introduce a new notion of
well-separation to reduce the work and space of our algorithm for
HDBSCAN$^*$.
\iffullversion
We also present a parallel approximate algorithm for OPTICS based on a recent sequential algorithm by Gan and Tao.
Finally, we give a new parallel divide-and-conquer algorithm for computing the dendrogram
and reachability plots, which are used in visualizing clusters of different
scale that arise for both EMST and HDBSCAN$^*$.
\else
We also give a new parallel divide-and-conquer algorithm for computing the dendrogram
and reachability plots, which are used in visualizing clusters of different
scale that arise for both EMST and HDBSCAN$^*$.
\fi
We show that our algorithms are theoretically efficient: they 
have work (number of operations) matching their sequential
counterparts, and polylogarithmic depth (parallel time).

We implement our algorithms and propose a memory
optimization that requires only a subset of well-separated pairs to
be computed and materialized, leading to savings in both space (up to 10x) and
time (up to 8x).  Our experiments on large real-world and synthetic data sets
using a 48-core machine show that our fastest algorithms outperform the best
serial algorithms for the problems by 11.13--55.89x, and existing parallel
algorithms by at least an order of magnitude.

\end{abstract}

%% file: intro.tex
\section{Introduction}\label{sec:intro}

This paper studies the two related geometric problems of Euclidean
minimum spanning tree (EMST) and hierarchical density-based spatial
clustering with added noise~\cite{Campello2015}. The problems
take as input a set of $n$ points in a $d$-dimensional space. EMST
computes a minimum spanning tree on a complete graph formed among the
$n$ points with edges between two points having the weight equal to their
Euclidean distance. EMST has many applications, including in
single-linkage clustering~\cite{gower1969minimum}, network placement
optimization~\cite{Wan2002}, and approximating the Euclidean traveling
salesman problem~\cite{Vazirani2010}.

\emph{Hierarchical density-based spatial clustering of applications with noise}
  (\emph{\hdbscan}) is a popular hierarchical clustering
algorithm~\cite{Campello2015}. The goal of density-based spatial
clustering is to cluster points that are in dense regions and close
together in proximity.  One of the most widely-used density-based
spatial clustering methods is the \emph{density-based spatial
  clustering of applications with noise} (\emph{DBSCAN}) method by Ester et
al.~\cite{Ester1996}.  DBSCAN requires two parameters, $\epsilon$ and
$\minpts$, which determine what is considered ``close'' and ``dense'',
respectively.  In practice, $\minpts$ is usually fixed to a small
constant, but many different values of $\epsilon$ need to be explored
in order to find high-quality clusters. Many efficient DBSCAN
algorithms have been designed both for the
sequential~\cite{BergGR17,GanT17,Gunawan13,Chen05} and the parallel
context (both shared memory and distributed memory)~\cite{wang2019dbscan,Song2018,Lulli2016,Hu2017,Patwary12,Gotz2015}.
To avoid repeatedly executing
DBSCAN for different values of $\epsilon$, the
OPTICS~\cite{Ankerst1999} and \hdbscan~\cite{Campello2015} algorithms
have been
proposed for constructing DBSCAN clustering hierarchies, from which clusters from
different values of $\epsilon$ can be generated. These algorithms are known to be robust to outliers in the
data set.  The algorithms are based on generating a minimum spanning
tree on the input points, where a subset of the edge weights are
determined by Euclidean distance and the remaining edge weights are
determined by a DBSCAN-specific metric known as the core distance (to
be defined in Section~\ref{sec:prelims}). Thus, the algorithms bear
some similarity to EMST algorithms.

There has been a significant amount of theoretical work on designing
fast sequential EMST algorithms
(e.g.,~\cite{Agarwal1991,Yao1982,Shamos1975Closest,Arya2016,CallahanK93}). There
have also been some practical implementations of
EMST~\cite{march2010fast,ChatterjeeCK10,Bentley78,NarasimhanZ01},
although most of them are sequential (part of the algorithm by
Chatterjee et al.~\cite{ChatterjeeCK10} is parallel).  The
state-of-the-art EMST implementations are either based on
generating a well-separated pair decomposition
(WSPD)~\cite{CallahanK95} and applying Kruskal's minimum spanning tree (MST) algorithm on
edges produced by the WSPD~\cite{ChatterjeeCK10,NarasimhanZ01}, or
dual-tree traversals on $k$-d trees integrated into Boruvka's
MST algorithm~\cite{march2010fast}. 
Much less work has been proposed for parallel \hdbscan and OPTICS~\cite{Patwary2013,santos2019hdbscanmapreduce}.
In this paper, we design new algorithms for EMST, which can also be leveraged to design a fast parallel \hdbscan algorithm.

This paper presents practical parallel in-memory algorithms for EMST
and \hdbscan, and proves that the
theoretical work (number of operations) of our implementations matches their state-of-the-art
counterparts ($O(n^2)$), while having polylogarithmic depth.\footnote{The work is the total number of operations and depth
  (parallel time) is the length of the longest sequential dependence.} Our algorithms are based on finding the WSPD and then running Kruskal's algorithm on the WSPD edges.
For our \hdbscan algorithm, we propose a new notion of
well-separation to include the notion of core distances, which enables
us to improve the space usage and work of our algorithm.

Given the MST from the EMST or the \hdbscan problem, we provide an
algorithm to generate a dendrogram, which represents the hierarchy of
clusters in our data set. For EMST, this solves the single-linkage
clustering problem~\cite{gower1969minimum}, and for \hdbscan, this gives us a
dendrogram as well as a reachability plot~\cite{Campello2015}.  We
introduce a work-efficient\footnote{A \defn{work-efficient} parallel algorithm has a work bound that matches the best sequential algorithm for the problem.}
parallel divide-and-conquer algorithm
that first generates an Euler tour on the tree, splits the tree into
multiple subtrees, recursively generates the dendrogram for each
subtree, and glues the results back together. An in-order traversal of
the dendrogram gives the reachability plot. Our algorithm takes
$O(n\log n)$ work and $O(\log^2 n\log\log n)$ depth. Our parallel dendrogram algorithm is of independent interest, as it
can also be applied to generate dendrograms for other clustering problems.

We provide optimized parallel implementations of our EMST and
\hdbscan algorithms. We introduce a
memory optimization that avoids computing and
materializing many of the WSPD pairs, which significantly improves the
performance of our algorithm (up to 8x faster and 10x less space).
We also provide optimized implementations of $k$-d
trees, which our algorithms use for spatial queries.

We perform a comprehensive set of experiments on both synthetic and
real-world data sets using varying parameters, and compare the
performance of our implementations to optimized sequential implementations as well as
existing parallel implementations.
Compared to existing EMST sequential implementations~\cite{march2010fast,mcinnes2017accelerated}, our fastest sequential
implementation is 0.89--4.17x faster (2.44x on average). 
On a 48-core machine with
hyper-threading, our EMST implementation achieves 14.61--55.89x speedup over the fastest sequential
implementations.
Our \hdbscan implementation  achieves
11.13--46.69x speedup over the fastest sequential implementations.
Compared to existing sequential and parallel implementations for \hdbscan~\cite{Gan2018, mcinnes2017accelerated,santos2019hdbscanmapreduce, Patwary2013}, our implementation is at least an order of magnitude faster.
Our source code is publicly available at {\url{https://github.com/wangyiqiu/hdbscan}}.

We summarize our contributions below:
\begin{itemize}[topsep=1pt,itemsep=0pt,parsep=0pt,leftmargin=15pt]
    \item New parallel algorithms for EMST and \hdbscan with strong theoretical guarantees.
    \item A new definition of well-separation that computes the HDBSCAN* MST using asymptotically less space.
    \item Memory-optimized parallel implementations for EMST and \hdbscan that give significant space and time improvements.
    \item A new parallel algorithm for dendrogram construction.
    \item A comprehensive experimental study of the proposed methods.
\end{itemize}

\iffullversion
In the Appendix,
we present several
additional theoretical results: (1) an EMST algorithm with
subquadratic work and polylogarithmic depth based on a
subquadratic-work sequential algorithm by Callahan and
Kosaraju~\cite{CallahanK93}; (2) an \hdbscan algorithm for two
dimensions with $O(\minpts^2 \cdot n\log n)$ work, matching the sequential algorithm by Berg et
al.~\cite{BergGR17}, and $O(\minpts \cdot
\log^2 n)$ depth; and (3) a work-efficient parallel algorithm for
approximate OPTICS based on the sequential algorithm by
Gan and Tao~\cite{Gan2018}.

The rest of the paper is organized as
follows. Section~\ref{sec:prelims} introduces relevant definitions. Section~\ref{sec:wspd} presents our parallel well-separated
pair decomposition approach and uses it to obtain parallel algorithms
for the two problems.
Section~\ref{sec:dendro} presents our
parallel dendrogram construction algorithm, which can be used to
generate the single-linkage clustering, \hdbscan dendrogram, and
reachability plot. Section~\ref{sec:experiment} presents
experimental results.  We conclude in
Section~\ref{sec:conclusion}.
\fi

%% file: prelim.tex
\section{Preliminaries}\label{sec:prelims}

\subsection{Problem Definitions}

\myparagraph{EMST} The
\defn{Euclidean Minimum Spanning Tree (EMST)} problem takes
$n$ points $\mathcal{P}=\{p_1,\ldots, p_{n}\}$ and returns a minimum spanning
  tree (MST) of
the complete undirected Euclidean graph of $\mathcal{P}$.

\myparagraph{DBSCAN$^*$} The \defn{DBSCAN$^*$} (density-based spatial
clustering of applications with noise) problem takes as input $n$
points $\mathcal{P}=\{p_1,\ldots, p_n\}$, a distance function $d$, and
two parameters $\epsilon$ and $\minpts$~\cite{Ester1996,Campello2015}.
A point $p$ is a \defn{core point} if and only if $\left|\{p_i
~|~p_i\in \mathcal{P},d(p,p_i)\le \epsilon\}\right|\ge \minpts$. A
point is called a \defn{noise point} otherwise.  We denote the set of
core points as $\mathcal{P}_{\codevar{core}}$.  DBSCAN$^*$ computes a partition of $\mathcal{P}_{\codevar{core}}$, where each subset is referred to as a \defn{cluster}, and also returns the remaining points as noise points.
Two points $p,q\in \mathcal{P}_{\codevar{core}}$ are in
the same cluster if and only if there exists a list of points
$p=\bar{p}_1, \bar{p}_2, \ldots, \bar{p}_{k-1}, \bar{p}_k=q$ in
$\mathcal{P}_{\codevar{core}}$ such that
$d(\bar{p}_{i-1},\bar{p}_{i})\le \epsilon$ for all $1<i\leq k$.
For a given set of points
and two parameters $\epsilon$ and $\minpts$, the clusters returned are
unique.\footnote{The original DBSCAN definition includes the notion of
  border points, which are non-core points that are within a distance
  of $\epsilon$ to core points~\cite{Ester1996}. DBSCAN$^*$ chooses to
  omit this to be more consistent with a statistical interpretation of
  clusters~\cite{Campello2015}.}

\myparagraph{\hdbscan} The \defn{\hdbscan} (hierarchical DBSCAN$^*$)
problem~\cite{Campello2015} takes the same input as DBSCAN$^*$, but
without the $\epsilon$ parameter, and computes a hierarchy of
DBSCAN$^*$ clusters for all possible values of $\epsilon$. We first
introduce some definitions, and then describe how the hierarchy is
computed and represented. The \defn{core distance} of a point $p$,
$\CD(p)$, is the distance from $p$ to its $\minpts$-nearest neighbor
(including $p$ itself).  The \defn{mutual reachability distance}
between two points $p$ and $q$ is defined to be $d_m(p,q) =
max\{\CD(p),\CD(q),d(p,q)\}$.
The \defn{mutual reachability graph} $G_\codevar{MR}$ is a complete undirected
graph, where the vertices are the points in $\mathcal{P}$, and the
edges are weighted by the mutual reachability distances.\footnote{The
related OPTICS problem also generates a hierarchy of clusters but
with a definition of reachability distance that is asymmetric,
leading to a directed graph~\cite{Ankerst1999}.}

The \hdbscan hierarchy is sequentially computed in two steps~\cite{Campello2015}.
The first step computes an MST of $G_\codevar{MR}$ and then adds a self-edge to each vertex weighted by its core distance. An example MST is shown in Figure~\ref{fig:hdbscan}a.
We note that the \hdbscan MST with $\minpts=1$ is equivalent to the EMST, since the mutual reachability distance at $\minpts=1$ is equivalent to the Euclidean distance.
\iffullversion
We further elaborate on the relationship between \hdbscan and EMST in Appendix~\ref{sec:minpts_ex}.
\fi
A dendrogram representing clusters at different values of $\epsilon$
is computed by removing edges from the MST plus self-edges graph in
decreasing order of weight.  The root of the dendrogram is a cluster
containing all points. Each non-self-edge removal splits a cluster
into two, which become the two children of the cluster in the
dendrogram.
The height of the split cluster in the dendrogram
is equal to the weight of the removed edge.
If the removed edge is a self-edge, we mark the
component (point) as a noise point.  An example of a dendrogram is
shown in Figure~\ref{fig:hdbscan}b.  If we want to return the clusters
for a particular value of $\epsilon$, we can horizontally cut the
dendrogram at that value of $\epsilon$ and return the 
resulting subtrees below the cut as the clusters or noise points. This is equivalent to removing edges
from the MST of $G_\codevar{MR}$ with weight greater than $\epsilon$.

For \hdbscan, the reachability plot (OPTICS sequence)~\cite{Ankerst1999} contains all
points in $\mathcal{P}$ in some order $\{p_i |\ i = 1, \dots, n\}$,
where each point $p_i$ is represented as a bar with height $\min \{d_m(p_i,
p_j)\ |\ j < i\}$. For \hdbscan, the order of the points is the order that they are visited in an execution of Prim's algorithm
on the MST of $G_\codevar{MR}$ starting from an
arbitrary point~\cite{Ankerst1999}. An example is shown in
Figure~\ref{fig:hdbscan}c.  Intuitively, the "valleys" of the
reachability plot correspond to clusters~\cite{Campello2015}.

\begin{figure}
  \vspace{-5pt}
  \includegraphics[trim={0 0 0 40},clip,width = 0.85\columnwidth]{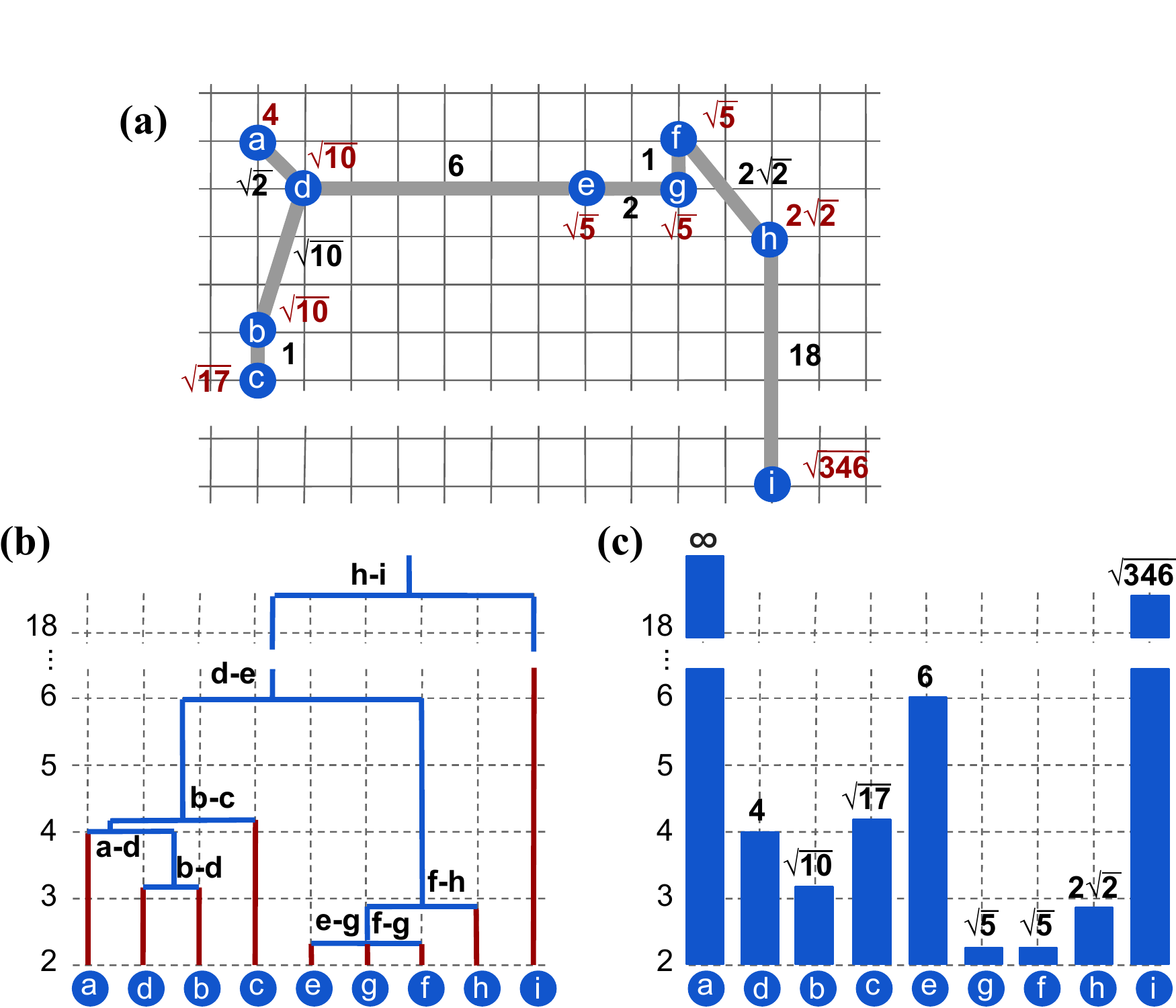}
  \caption{(a) An MST of the \hdbscan mutual reachability graph on an example data set in 2D. The red number
  next to each point is the core distance of the point for $\minpts=3$.
  The Euclidean distances between points are denoted by grey edges, whose values are marked in black.
  For example, $a$'s core distance is 4 because $b$ is $a$'s third nearest neighbor (including itself) and $d(a,b) = 4$.
  The edge weight of $(a,d)$ is $\max\{4,\sqrt{10},\sqrt{2}\} = 4$.
  (b) An \hdbscan dendrogram for the data set.
  A point becomes a noise point when its vertical line becomes red. For example, if we
  cut the dendrogram at $\epsilon=3.5$, then we have two clusters $\{d,b\}$ and $\{e,g,f,h\}$, while
  $a$, $c$ and $i$ are noise points.
  (c) A reachability plot for the data set starting at point $a$. The two ``valleys'',
  $\{a,b,c,d\}$  and $\{e,f,g,h\}$, are the two most obvious clusters.
    }
  \label{fig:hdbscan}
\end{figure}

\subsection{Parallel Primitives}
We use the classic
\defn{work-depth model} for analyzing parallel shared-memory
algorithms~\cite{CLRS,KarpR90,JaJa92}.
The \defn{work} $W$ of an algorithm is the number of instructions in
the computation, and the \defn{depth} $D$ is the longest sequential
dependence.  Using Brent's scheduling theorem~\cite{Brent1974},  we can
execute a parallel computation in $W/p+D$ running time using $p$ processors.
In practice, we can use randomized work-stealing schedulers that are available in existing languages such as Cilk, TBB, X10, and Java Fork-Join.
We assume that
priority concurrent writes are supported in $O(1)$ work and depth.

\defn{Prefix sum} takes as input a sequence $[a_1, a_2, \ldots , a_{n}]$, an
associative binary operator $\oplus$, and an identity element $i$, and returns
the sequence $[i, a_1, (a_1 \oplus a_2), \ldots , (a_1 \oplus a_2
  \oplus \ldots \oplus a_{n-1})]$
as well as the overall sum (using binary operator $\oplus$) of the
elements.  
\defn{Filter} takes an array $A$ and a
predicate function $f$, and returns a new array containing $a \in A$
for which $f(a)$ is true, in the same order that they appear in $A$. Filter can be implemented using prefix sum.
\defn{Split} takes an array $A$ and a predicate
function $f$, and moves all of the ``true'' elements before the
``false'' elements.
Split can be implemented using filter.
The \defn{Euler tour} of a tree takes as input an
adjacency list representation of the tree and returns a directed
circuit that traverses every edge of the tree exactly once.  \defn{List ranking} takes a linked list
with values on each node and returns for each node the sum of values
from the node to the end of the list.
All of
the above primitives can be implemented in $O(n)$ work and $O(\log n)$
depth~\cite{JaJa92}.
\defn{Semisort}~\cite{gu15semisort} takes as input $n$ items, each with a key, and groups the items with the same key together, without any guarantee on the ordering of items with different keys. 
This algorithm takes $O(n)$ expected work and $O(\log n)$ depth with high probability.
A parallel \defn{hash table} supports $n$ inserts, deletes, and finds in $O(n)$ work and $O(\log n)$ depth with high probability~\cite{Gil91a}.
\textsc{WriteMin} is a priority concurrent write that takes as input two
arguments, where the first argument is the location to write to and the
second argument is the value to write; on concurrent writes, 
the smallest value is written~\cite{ShunBFG2013}.  

\subsection{Relevant Techniques}

\myparagraph{$k$-NN Query} A \defn{$k$-nearest neighbor ($k$-NN)
  query} takes a point data set $\mathcal{P}$ and a distance function,
and returns for each point in $\mathcal{P}$ its $k$ nearest neighbors
(including itself).  Callahan and Kosaraju~\cite{callahan1993optimal}
show that $k$-NN queries in Euclidean space for all points can be
solved in parallel in $O(kn\log n)$ work and $O(\log n)$ depth.

\myparagraph{\kdt{}} A \defn{\kdt{}} is a commonly used data structure for $k$-NN
queries~\cite{friedman1977kdtree}.  It is a binary tree that is
constructed recursively: each node in the tree represents a
set of points, which are partitioned between its two children by
splitting along one of the dimensions; this process is recursively
applied on each of its two children until a leaf node is reached (a
leaf node is one that contains at most $c$ points, for a predetermined
constant $c$). It can be constructed in parallel by processing each
child in parallel. A $k$-NN query can be answered by traversing nodes
in the tree that are close to the input point, and pruning nodes
further away that cannot possibly contain the $k$ nearest neighbors.

\myparagraph{$\bccp$ and $\bccpstar$} Existing algorithms, as well as
some of our new algorithms, use subroutines for solving the
\defn{bichromatic closest pair ($\bccp$)} problem, which takes as
input two sets of points, $A$ and $B$, and returns the pair of points
$p_1$ and $p_2$ with minimum distance between them, where $p_1 \in A$
and $p_2 \in B$. 
We also define a variant,
the \defn{$\bccpstar$} problem, that finds the pair of points with the
minimum mutual reachability distance, as defined for \hdbscan.

\myparagraph{Well-Separated Pair Decomposition} We use the same
definitions and notations as in Callahan and
Kosaraju~\cite{CallahanK95}.  Two sets of points, $A$ and $B$, are
\defn{well-separated} if $A$ and $B$ can each be contained in spheres of
radius $r$, and the minimum distance between the two spheres 
is at least $sr$, for a \defn{separation constant} $s$ (we use $s=2$ throughout the paper). 
An \defn{interaction product} of point sets $A$ and $B$ is defined to be
$A \otimes B = \{\{p,p'\}|\ p \in A,\ p' \in B,\ p\neq p'\}$.  The
set $\{\{A_1,B_1\},\ldots,\{A_k,B_k\}\}$ is a
\defn{well-separated realization} of $A \otimes B$ if:
\textbf{(1)} $A_i \subseteq A$ and $B_i \subseteq B$ for all $i=1,...,k$; \textbf{(2)} $A_i \cap B_i = \emptyset$ for all $i=1,...,k$; \textbf{(3)} $(A_i \otimes B_i) \ \bigcap\ (A_j \otimes B_j) = \emptyset$ for all $i,j$ where $1 \leq i < j \leq k$; \textbf{(4)} $A \otimes B = \bigcup_{i=1}^k A_i \otimes B_i$; \textbf{(5)} $A_i$ and $B_i$ are well-separated for all $i=1,...,k$.

For a point set $\mathcal{P}$, a \defn{well-separated pair decomposition (WSPD)} is
a well-separated realization of $\mathcal{P} \otimes \mathcal{P}$.
We discuss how to construct a WSPD using a \kdt in Section~\ref{sec:wspd}.

\myparagraph{Notation} Table~\ref{table:notation} shows notation frequently used in the paper. 

\begin{table}[]
\footnotesize
  \begin{tabular}{|c | l |}
\hline
\textbf{Notation} & \multicolumn{1}{c|}{\textbf{Definition}} \\ \hline
$d(p,q)$ & Euclidean distance between points $p$ and $q$. \\ \hline
$d_m(p,q)$ & Mutual reachability distance between points $p$ and $q$. \\ \hline
$d(A,B)$ & \makecell[cl]{Minimum distance between the bounding spheres of \\ points in tree node $A$ and points in tree node $B$. }\\ \hline
$w(u,v)$ & Weight of edge $(u,v)$. \\ \hline
$\field{A}{diam}$ & \makecell[cl]{Diameter of the bounding sphere of points in tree node $A$.}\\ \hline
$\CDmin(A)$ & Minimum core distance of points in tree node $A$. \\ \hline
$\CDmax(A)$ & Maximum core distance of points in tree node $A$. \\ \hline
\end{tabular}
  \caption{Summary of Notation}\label{table:notation}  
\end{table}

%% file: hdbscan_wspd.tex
\section{Parallel EMST and \hdbscan}\label{sec:wspd}

In this section, we present our new parallel algorithms for EMST and \hdbscan. We also introduce our new memory optimization to improve space usage and performance in practice.
\input{emst_wspd}

\subsection{\hdbscan} \label{sec:hdbscanwspd}

\subsubsection{Baseline} \label{sec:hdbscan_wspd:existing}

Inspired by a previous sequential approximate algorithm to solve the OPTICS problem
by Gan and Tao~\cite{Gan2018}, we  modified and parallelized their algorithm
to compute the exact \hdbscan as our baseline.
First, we perform \knn{} queries using Euclidean distance
with $k = \minpts$ to compute the core distances.
Gan and Tao's original algorithm creates a mutual reachability graph of size $O(n\cdot \minpts^2)$, using an approximate notion of $\bccp$ between each WSPD pair,
and then computes its MST using Prim's algorithm.
Our exact algorithm parallelizes their algorithm, and instead uses the exact $\bccpstar$ computations based on the mutual reachability distance to form the mutual reachability graph.
In addition, we also compute the MST on the generated edges using the MemoGFK optimization described in Section~\ref{sec:impl:memogfk}.
Summed across all well-separated pairs,
the $\bccp$ computations take quadratic work and constant depth.
Therefore, our baseline algorithm takes $O(n^2)$ work and $O(\log^2 n)$ depth, and computes the exact \hdbscan.
\iffullversion
In Section~\ref{section:approx} of the Appendix,
we also describe a work-efficient parallel approximate algorithm based on~\cite{Gan2018}.
\fi

\subsubsection{Improved Algorithm}\label{sec:hdbscan_wspd:our}
Here, 
we present a more space-efficient 
algorithm that is also faster in practice. The idea is to use a
different definition of well-separation for the WSPD in
\hdbscan.
We denote the maximum and minimum core distances of the points in node $A$ as $\CDmax(A)$ and $\CDmin(A)$, respectively.
Consider a pair $(A,B)$ in the WSPD. We
define $A$ and $B$ to be \defn{geometrically-separated} if
$d(A,B)\geq \max\{ \field{A}{diam}, \field{B}{diam} \}$ and
\defn{mutually-unreachable} if $\max\{\allowbreak d(A,B),\CDmin(A),\CDmin(B)\}
\geq \max\{\field{A}{diam},
\field{B}{diam},\CDmax(A),\allowbreak\CDmax(B)\}$. We consider $A$ and $B$ to
be well-separated if they are geometrically-separated,
mutually-unreachable, or both.
Note that the original definition of well-separation with only includes the first condition. 

This leads to space savings because in \cref{alg:wspd}, recursive
calls to procedure \textsc{FindPair}$(A,B)$ on \cref{line:findpair}
will not terminate until $A$ and $B$ are well-separated.
Since our new definition is a disjunction between
mutual-unreachability and geometric-separation,
the calls to
\textsc{FindPair} can terminate earlier, leading to fewer pairs
generated.
When constructing the mutual reachability subgraph to pass to MST, we add only a single  edge between the
$\bccpstar$ ($\bccp$ with respect to mutual reachability distance) of each well-separated pair.
With our new definition, the total number of edges generated is upper bounded by the size of the WSPD, which is $O(n)$~\cite{CallahanK95}.
In contrast, Gan and Tao's approach generates $O(n\cdot \minpts^2)$ edges.

\begin{theorem}
Under the new definition of well-separation, our algorithm computes an MST of the mutual reachability graph.
\end{theorem}

\begin{proof}

Under our new definition, well-separation is defined as
the disjunction between being geometrically-separated and
mutually-unreachable.
We connect an edge between each well-separated pair $(A,B)$ with the mutual-reachability distance
$\max\{d(u^*,v^*),\allowbreak\CD(u^*),\allowbreak\CD(v^*)\}$ as the edge weight,
where $u^*\in A$, $v^*\in B$, and $(u^*,v^*)$ is the $\bccpstar$ of $(A,B)$.
We overload the notation $\bccpstar(A,B)$ to also denote
the mutual-reachability distance of $(u^*,v^*)$.

Consider the point set $\field{P}{root}$, which is contained in the root node of
the tree associated with its WSPD.
Let $T$ be the MST of the full mutual reachability graph $G_{MR}$.
Let $T'$ be the MST of the mutual reachability subgraph $G'_{MR}$, computed by
connecting the $\bccpstar$ of each well-separated pair.
To ensure that $T'$ produces the correct \hdbscan clustering,
we prove that it has the same weight as $T$---in other words, $T'$
is a valid MST of $G_{MR}$.

We prove the optimality of $T'$ by induction on each tree node $P$.
Since the WSPD is hierarchical, each node $P$ also has a valid WSPD consisting of
a subset of pairs of the WSPD of $\field{P}{root}$.
Let $(u,v)$ be an edge in $T$. There exists an edge $(u',v')\in T'$ that connects
the same two components as in $T$ if we were to remove $(u,v)$.
We call $(u',v')$ the \defn{replacement}
of $(u,v)$, which is \defn{optimal} if $w(u',v')=w(u,v)$.
Let $T_P$ and $T'_P$  be subgraphs of $T$ and $T'$, respectively, containing points in $P$, but not necessarily spanning $P$.
We inductively hypothesize that all edges of $T'_P$ are optimal.
In the base case, a singleton tree node $P$ satisfies the hypothesis by having no edges.

Now consider any node $P$ and edge $(u,v)\in T_P$. The children of $P$ are optimal by our inductive hypothesis.
We now prove that the edges connecting the children of $P$ are optimal. 
Points $u$ and $v$ must be from a well-separated pair $(A,B)$,
where $A$ and $B$ are children of $P$ from the WSPD hierarchy.
Let $U$ and $V$ be a partition of $P$ formed by a cut in $T_P$
that separates point pair $(u,v)$, where $u\in U$ and $v \in V$. We want to prove that the replacement of $(u,v)$ in $T_P'$ is optimal.

We now discuss the first scenario of the proof, shown in Figure~\ref{fig:dwspd_proof}a,
where the replacement edge between $U$ and $V$ is
$(u',v')=\bccpstar(A,B)=(u^*,v^*)$, and we assume without loss of generality
that $u'\in A\cap U$ and $v'\in B\cap V$.
Since $(u,v)$ is the closest pair of points connecting $U$ and $V$ by the cut property, then $(u',v')$, the $\bccpstar$ of $(A,B)$, must be optimal; otherwise,
$(u,v)$ has smaller weight than $\bccpstar(A,B)$, which is a contradiction.
This scenario easily generalizes to the case where $A$ and $B$ happen to be completely within
$U$ and $V$, respectively.

\begin{figure}[t]
  \includegraphics[trim=0 0 0 100,clip,width=0.5\linewidth]{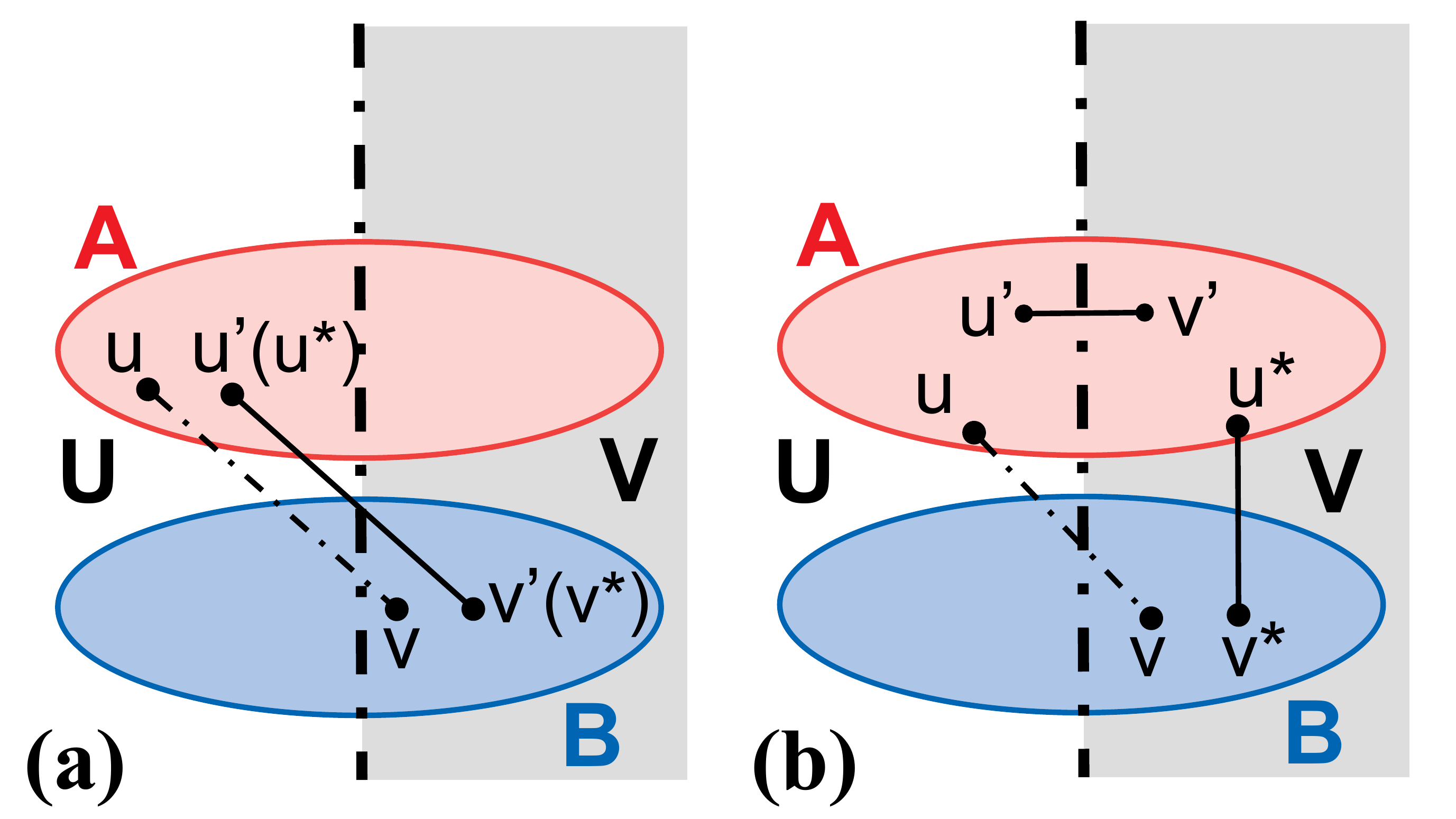}
  \caption{In this figure, we show the two proof cases for \hdbscan. We use an oval to represent each node in the WSPD, and solid black dots to represent data points. We represent the partition of the space to $U$ and $V$ using a cut represented by a dotted line.} \label{fig:dwspd_proof}
\end{figure}

We now discuss the second scenario, shown in Figure~\ref{fig:dwspd_proof}b,
where $\bccpstar(A,B)=(u^*,v^*)$ is internal to either $U$ or $V$.  We
assume without loss of generality that $u^*\in A\cap V$ and $v^*\in B\cap V$, and
that  $U$ and $V$ are connected by some intra-node edge $(u',v')$ of $A$ in $T'_P$. 
We want to prove that $(u',v')$ is an optimal replacement edge.
We consider two cases based on the relationship between $A$ and $B$ under our new definition of well-separation.

\myparagraph{Case 1}
Nodes $A$ and $B$ are mutually-unreachable, and may or may not
be geometrically-separated.
The weight of $(u',v')$ is
$\max\{d(u',v'),\allowbreak\CD(u'),\CD(v')\}\leq \max\{\field{A}{diam},\CDmax(A)\}$.
Consider the $\bccpstar$ pair $(u^*,v^*)$ between $A$ and $B$. Based on the fact that
$A$ and $B$ are mutually-unreachable, we have
\begin{align*}
  \bccpstar(A,B)&=\max\{d(u^*,v^*),\CD(u^*),\CD(v^*)\}\\
  &\geq \max\{d(A,B),\CDmin(A),\CDmin(B)\}\\
  &\geq \max\{\field{A}{diam}, \field{B}{diam}, \CDmax(A), \CDmax(B)\}\\
  &\geq \max\{\field{A}{diam},\CDmax(A)\},
\end{align*}
where the inequality from the second to the third line above comes from the definition of mutual-unreachability.
Therefore, $w(u',v')$ is not larger than $\bccpstar(A,B)=w(u^*,v^*)$,
and by definition of $\bccpstar$, $w(u^*,v^*)$ is not larger than $w(u,v)$.
Hence, $w(u',v')$ is not larger than $w(u,v)$.
On the other hand, $w(u',v')$ is not smaller than $w(u,v)$, since otherwise
we could form a spanning tree with a smaller weight than $T_P$, contradicting
the fact that it is an MST.
Thus, $(u',v')$ is optimal.

\myparagraph{Case 2}
Nodes $A$ and $B$ are geometrically-separated and not mutually-unreachable.
By the definition of $\bccpstar$, we know that $w(u^*,v^*)\leq w(u,v)$, which implies 
\begin{align*}
  \max\{\CD(u^*),\CD(v^*),d(u^*,v^*)\}\leq \max\{\CD(u),\CD(v),d(u,v)\}\\
  \max\{\CD(u^*),\CD(u),d(u,u^*)\}\leq \max\{\CD(u),\CD(v),d(u,v)\}.
\end{align*}
To obtain the second inequality above from the first, we replace $\CD(v^*)$ on the left-hand side with $\CD(u)$, since $\CD(u)$ is also on the right-hand side; we also replace $d(u^*,v^*)$ with $d(u,u^*)$ because of the geometric separation of $A$ and $B$.
Since $(u',v')$ is the lightest \bccpstar{} edge of some well-separated pair in $A$, $\max\{\CD(u'),\CD(v'),\allowbreak d(u',v')\} \leq \max\{\CD(u),\CD(u^*),d(u,u^*)\}$. We then have
\begin{align*}
  \max\{\CD(u'),\CD(v'),d(u',v')\}\
  \leq \allowbreak \max\{\CD(u),\CD(v),d(u,v)\}.
\end{align*}
This implies that $w(u',v')$ is not larger than $w(u,v)$.
Since $(u,v)$ is an edge of MST $T_P$, the weight of the replacement edge $w(u',v')$ is also not smaller than $w(u,v)$, and
hence $(u',v')$ is optimal.

Case 1 and 2 combined prove the optimality
of replacement edges in the second scenario.
Considering both scenarios,
we have shown that each replacement edge in $T_p'$ connecting the children of $P$ is optimal, which proves the inductive hypothesis.
Applying the inductive hypothesis to $\field{P}{root}$ completes the proof.
\end{proof}

Our algorithm achieves the following bounds.

\vspace{-3pt}
\begin{theorem}
  Given a set of $n$ points, we can compute the MST on the mutual
  reachability graph in $O(n^2)$ work, $O(\log^2 n)$ depth, and $O(n\cdot \minpts)$
  space. 
\end{theorem}
\vspace{-3pt}

\begin{proof}
Compared to the cost of GFK for EMST, GFK for \hdbscan has the additional cost of computing the core distances,
which takes $O(\minpts \cdot n \log n)$ work and
$O(\log n)$ depth using \knn~\cite{callahan1993optimal}.
With our new definition of well-separation, the WSPD computation will
only terminate earlier than in the original definition, and so the
bounds that we showed for EMST above still hold.  The new WSPD
definition also gives an $O(n)$ space bound for the well-separated
pairs.  The space usage of the \knn computation is $O(n\cdot
\minpts)$, which dominates the space usage.
Overall, this gives
$O(n^2)$ work, $O(\log^2 n)$ depth, and $O(n\cdot \minpts)$ space.
\end{proof}

Our algorithm gives a clear improvement in space usage over the naive
approach of computing an MST from the mutual reachability graph,  which takes $O(n^2)$ space, and our parallelization of the
exact version of Gan and Tao's algorithm, which takes $O(n \cdot
\minpts^2)$ space. We will also see that the smaller memory footprint
of this algorithm leads to better performance in practice.

\subsubsection{Implementation}
\iffullversion
We implement three algorithms for \hdbscan:
a parallel version of the approximate algorithm based
on Gan and Tao~\cite{Gan2018}, a parallel exact algorithm based
on Gan and Tao, and our space-efficient algorithm from Section~\ref{sec:hdbscan_wspd:our}.
Our implementations all use Kruskal's algorithm for MST and use the memory optimization introduced for
MemoGFK in Section~\ref{sec:impl:memogfk}.
\else
We implement two algorithms for \hdbscan:
a parallel exact algorithm based on Gan and Tao~\cite{Gan2018},
and our space-efficient algorithm from Section~\ref{sec:hdbscan_wspd:our}.
Our implementations both use Kruskal's algorithm for MST and use the memory optimization introduced for
MemoGFK in Section~\ref{sec:impl:memogfk}.
\fi
For our space-efficient
algorithm, we modify the WSPD and MemoGFK algorithm to use
our new definition of well-separation.

%% file: emst_wspd.tex
\subsection{EMST}\label{sec:emst_wspd}

To solve EMST, Callahan and Kosaraju present an algorithm for constructing a WSPD that 
creates an edge between the $\bccp$ of each pair in the WSPD with weight equal to their distance, and then runs an MST algorithm on these edges. They show that their algorithm takes $O(T_d(n,n) \log n)$ work~\cite{CallahanK93}, where $T_d(n,n)$
refers to the work of computing $\bccp$ on two sets each of size $n$.

For our parallel EMST algorithm, we parallelize WSPD construction algorithm, 
and then develop a parallel variant of Kruskal's MST algorithm that runs on the edges formed by the pairs in the WSPD. We also propose a non-trivial optimization to make the
implementation fast and memory-efficient.

\subsubsection{Constructing a WSPD in Parallel}\label{sec:wspd-impl}

We introduce the basic parallel WSPD in Algorithm~\ref{alg:wspd}.
Prior to calling WSPD, we construct a spatial median \kdt{} $T$ in parallel
with each leaf containing one point.
Then, we call the procedure \wspd{} on \cref{line:wspd} and make the root node of
$T$ its input.  In \textsc{Wspd}, we make parallel calls to \textsc{FindPair}
on the two children of all non-leaf nodes by recursively calling \textsc{Wspd}.
The procedure
\textsc{FindPair} on \cref{line:findpair} takes as input a pair
$(P,P')$ of nodes in $T$, and checks whether $P$ and $P'$ are well-separated. If they are well-separated, then the algorithm
records them as a well-separated pair on \cref{line:wspd-check}; otherwise, the algorithm
splits the set with the larger bounding sphere into its two children
and makes two recursive calls in parallel
(Lines~\ref{line:wspd-spawn-start}--\ref{line:wspd-spawn-end}). This
process is applied recursively until the input pairs are
well-separated.
The major difference of Algorithm 1 from the serial version is the parallel thread-spawning on Lines 3--5 and 12--14. This procedure generates a WSPD with $O(n)$ pairs~\cite{CallahanK93}.

\begin{algorithm}[!t]
  \caption{Well-Separated Pair Decomposition}\label{alg:wspd}
\begin{algorithmic}[1]
  \Procedure{Wspd}{$A$} \label{line:wspd}
  \If {$|A| > 1$}
  \InParallel
  \State \Call{Wspd}{$A_{\codevar{left}}$} \Comment{parallel call on the left child of $A$}
  \State \Call{Wspd}{$A_{\codevar{right}}$} \Comment{parallel call on the right child of $A$}
  \EndParallel
  \State \Call{FindPair}{$A_{\codevar{left}}$, $A_{\codevar{right}}$}
  \EndIf
  \EndProcedure

  \Procedure{FindPair}{$P$, $P'$} \label{line:findpair}
  \If {$\field{P}{diam} < \field{P'}{diam}$}
  \State \Call{Swap}{$P$, $P'$}
  \EndIf

  \If {\Call{WellSeparated}{$P$, $P'$}} 
  \Call{Record}{$P$, $P'$}\label{line:wspd-check}
  \Else
  \InParallel
  \State \Call{FindPair}{$P_{\codevar{left}}$, $P'$} \Comment{$P_{\codevar{left}}$ is the left child of $P$} \label{line:wspd-spawn-start}
  \State \Call{FindPair}{$P_{\codevar{right}}$, $P'$} \Comment{$P_{\codevar{right}}$ is the right child of $P$} \label{line:wspd-spawn-end}
  \EndParallel
  \EndIf
  \EndProcedure
\end{algorithmic}
\end{algorithm}

\begin{algorithm}[!t]
\caption{Parallel GeoFilterKruskal}\label{alg:gfk}
\begin{algorithmic}[1]
  \Procedure{ParallelGFK}{WSPD: $S$, Edges: $E_{\codevar{out}}$, UnionFind: $\codevar{UF}$}\label{line:gfk}

  \State $\beta = 2$
  \While{$|E_{\codevar{out}}| < (n-1)$}

  \State $(S_l,S_u)=$ \Call{Split}{$S$, $f_\beta$} \Comment{For a pair $(A,B)$, $f_\beta$ checks if $|A|+|B|\leq\beta$} \label{line:betasplit}
  \State $\rho_{\codevar{hi}} = \min_{(A,B)\in S_u}d(A,B)$ \label{line:getrho}
  \State $(S_{l1},S_{l2}) = $ \Call{Split}{$S_l$, $f_{\rho_{\codevar{hi}}}$}  \Comment{For a pair $(A,B)$, $f_{\rho_{\codevar{hi}}}$ checks if $\bccp(A,B)\leq\rho_{\codevar{hi}}$} \label{line:rhosplit}
\State $E_{l1} = $ \Call{GetEdges}{$S_{l1}$} \Comment{Retrieves edges associated with pairs in $S_{l1}$}\label{line:getedges}
  \State \Call{ParallelKruskal}{$E_{l1}$, $E_{\codevar{out}}$, $\codevar{UF}$}\label{line:kruskal}
  \State $S=$ \Call{Filter}{$S_{l2}\cup S_u$, $f_{\codevar{diff}}$} \Comment{For a pair $(A,B)$, $f_{\codevar{diff}}$ checks  points in $A$ are in different component from $B$ in $\codevar{UF}$}\label{line:filtering}
  \State $\beta=\beta\times 2$
  \EndWhile
  \EndProcedure
\end{algorithmic}
\end{algorithm}

\input{wspd_ex}

\subsubsection{Parallel GFK Algorithm for EMST}\label{sec:gfk}

The original algorithm by Callahan and Kosaraju~\cite{CallahanK93}
  computes the $\bccp$ between each pair in the WSPD to
generate a graph from which an MST can be computed to obtain the EMST.
However, it is not necessary to compute the $\bccp$ for all pairs, as
observed by Chatterjee et al.~\cite{ChatterjeeCK10}.  Our
implementation only computes the $\bccp$ between a pair if their
points are not yet connected in the spanning forest generated so far.
This optimization reduces the total number of $\bccp$
calls. Furthermore, we propose a memory optimization that avoids
materializing all of the pairs in the WSPD. We will first describe how
we obtain the EMST from the WSPD, and then give details of our memory
optimization.

The original Kruskal's algorithm is an MST algorithm that takes
input edges sorted by non-decreasing weight, and processes the edges
in order, using a union-find data structure to join components for
edges with endpoints in different components.  Our implementation is
inspired by a variant of Kruskal's algorithm, GeoFilterKruskal (GFK).
This algorithm was used for sequential EMST by Chatterjee et
al.~\cite{ChatterjeeCK10}, and for MST in general graphs by Osipov et
al.~\cite{Osipov2009}.  It improves Kruskal's algorithm by avoiding
the $\bccp$ computation between pairs unless needed, and prioritizing
$\bccp$s between pairs with smaller cardinalities, which are cheaper,
with the goal of pruning more expensive $\bccp$ computations.

We propose a parallel GFK algorithm as shown in Algorithm~\ref{alg:gfk}.
It uses Kruskal's MST
algorithm as a subroutine by passing it batches of edges, where each
batch has edges with weights no less than those of edges in previous batches, and the union-find structure is shared across multiple invocations of Kruskal's algorithm.
\textsc{ParallelGFK} takes as input the WSPD
pairs $S$, an array $E_{\codevar{out}}$ to store the MST edges, and a union-find
structure $\codevar{UF}$.
On each round, given a constant $\beta$, we only consider node pairs in the WSPD with cardinality (sum of sizes)
at most $\beta$ because it is cheaper to compute their $\bccp$s.
To do so, the set of pairs $S$ is partitioned into $S_l$, containing
pairs with cardinality at most $\beta$, and $S_u$, containing the
remaining pairs (\cref{line:betasplit}).
However, it is only correct
to consider pairs in $S_l$ that produce edges lighter than any of
the pairs in $S_u$. On \cref{line:getrho},  we compute an upper bound $\rho_{\codevar{hi}}$ for the edges in
$S_l$  by setting $\rho_{\codevar{hi}}$ equal to the minimum
$d(A,B)$ for all $(A,B)\in S_u$ (this is a lower bound on the edges weights formed by these pairs).
In the example shown in Figure~\ref{fig:wspd_ex},
in the first round, with $\beta=2$, 
 the set $S_l$ contains $(a,d)$, $(b,c)$, $(f,g)$, and $(e,h)$, and the set 
$S_u$ contains $(h,Q_7)$, $(e,Q_7)$, $(e,Q_2)$, $(Q_4,Q_5)$, $(Q_2,Q_6)$, and $(Q_1,i)$.
$\rho_{\codevar{hi}}$ corresponds to $(e,Q_7)$ on Line~\ref{line:getrho}.
Then, we compute the $\bccp$ of all
elements of set $S_l$, and split it into $S_{l1}$ and $S_{l2}$,
where $S_{l1}$ has edges with weight at most $\rho_{\codevar{hi}}$
(\cref{line:rhosplit}).
On Line~\ref{line:rhosplit},  $S_{l1}$ contains 
$(a,d)$, $(b,c)$ and $(f,g)$, as their $\bccp$ distances
are smaller than $\rho_{\codevar{hi}}=d(e,Q_7)$, and $S_{l2}$ contains $(e,h)$ .
After that, $E_{l1}$, the edges corresponding to
$S_{l1}$, are passed to Kruskal's algorithm (Lines~\ref{line:getedges}--\ref{line:kruskal}).  The remaining pairs
$S_{l2}\cup S_u$ are then filtered based on the result of Kruskal's
algorithm (Line~\ref{line:filtering})---in
particular, pairs that are connected in the union-find structure of
Kruskal's algorithm can be discarded, and for many of these pairs we never have to compute their
 $\bccp$.
In Figure~\ref{fig:wspd_ex}, the second round processes $(e,h)$, $(h,Q_7)$, $(e,Q_7)$, $(e,Q_2)$, $(Q_4,Q_5)$, $(Q_2,Q_6)$, and $(Q_1,i)$, and works similarly to Round~1. However, $(Q_2,Q_6)$ gets filtered out during the second round, and we never have to compute its $\bccp$, leading to less work compared to a naive algorithm. Finally, the subsequent rounds process a single pair $(Q_1,i)$.
At the end of each round, we double the value of $\beta$
to ensure that there are logarithmic number of rounds and hence better depth (in contrast, the sequential algorithm of Chatterjee et al.~\cite{ChatterjeeCK10} increases $\beta$ by 1 every round).
Throughout the algorithm, we cache the $\bccp$ results of pairs
to avoid repeated computations.
Overall, the main difference between Algorithm 2 and sequential algorithm is the use of parallel primitives on nearly every line of the pseudocode, and the exponentially increasing value of $\beta$ on Line 11, which is crucial for achieving a low depth bound.

The following theorem summarizes the bounds of our algorithm.

\begin{theorem}
We can compute the EMST on a set of $n$ points in constant dimensions in
$O(n^2)$ work and $O(\log^2 n)$ depth.
\end{theorem}

\begin{proof}
Callahan~\cite{callahan1993optimal} shows that a WSPD with $O(n)$ well-separated pairs can be computed in $O(n\log n)$ work and $O(\log n)$ depth, which we use for our analysis.
Our parallel GeoFilterKruskal algorithm for EMST proceeds in rounds, and processes the well-separated pairs
in an increasing order of cardinality. 
Since $\beta$ doubles on each round, there can be at most $O(\log n)$
rounds since the largest pair can contain $n$ points.
Within each round, the \textsc{Split} on Line~\ref{line:betasplit}
and \textsc{Filter} on Line~\ref{line:filtering} both take $O(n)$ work and $O(\log n)$ depth.
We can compute the $\bccp$ for each pair on Line~\ref{line:rhosplit}
by computing all possible point distances between the pair,
and using \textsc{WriteMin}  to obtain the minimum distance.
Since the $\bccp$ of each pair will only be computed once and
is cached, the total work of $\bccp$ on Line~\ref{line:rhosplit} is
$\sum_{A,B\in S}|A||B|=O(n^2)$ work since the WSPD is an exact set cover
for all distinct pairs.
Therefore, Line~\ref{line:rhosplit} takes $O(n^2)$ work across all 
rounds and $O(1)$ depth for each round.
Given $n$ edges, the MST computation on Line~\ref{line:kruskal}
can be done in $O(n\log n)$ work and $O(\log n)$
depth using existing parallel algorithms~\cite{JaJa92}.
Therefore, the overall work is $O(n^2)$.
Since each round takes $O(\log n)$ depth, and there are $O(\log n)$ rounds,
the overall depth is $O(\log^2 n)$.
\end{proof}

\iffullversion
We note that there exist subquadratic work $\bccp$
algorithms~\cite{Agarwal1991}, which result in a subquadratic work
EMST algorithm. Although the algorithm is impractical and no
implementations exist, for theoretical interest we give a
work-efficient parallel algorithm with polylogarithmic depth
in Section~\ref{sec:subquadratic-emst} of the Appendix.
\fi

We implemented our own sequential and parallel versions of the GFK
algorithm as a baseline based on Algorithm~\ref{alg:gfk}, which we
found to be faster than the implementation of Chatterjee et
al.~\cite{ChatterjeeCK10} in our experiments.  In addition, because
the original GFK algorithm requires materializing the full WSPD, its
memory consumption can be excessive, limiting the algorithm's
practicality.  This issue worsens as the dimensionality of the data set
increases, as the number of pairs in the WSPD increases exponentially
with the dimension.  While Chatterjee et al.~\cite{ChatterjeeCK10}
show that their GFK algorithm is efficient, they consider
much smaller data sets than the ones in this paper.

\subsubsection{The MemoGFK Optimization} \label{sec:impl:memogfk}

To tackle the memory consumption issue, we propose an optimization to
the GFK algorithm, which reduces its space usage and improves its
running time in practice.  We call the resulting algorithm
\defn{MemoGFK} (memory-optimized GFK).  The basic idea is that, rather
than materializing the full WSPD at the beginning, we partially
traverse the \kdt on each round and retrieve only the pairs that are
needed.  The pseudocode for our algorithm is shown in
Algorithm~\ref{alg:memogfk}, where \textsc{ParallelMemoGFK} takes in
the root $R$ of a \kdt{}, an array $E_{\codevar{out}}$ to store the
MST edges, and a union-find structure $\codevar{UF}$.

\begin{algorithm}[!t]
\caption{Parallel MemoGFK}\label{alg:memogfk}
\begin{algorithmic}[1]
  \Procedure{ParallelMemoGFK}{\kdt root: $R$, Edges: $E_{\codevar{out}}$, UnionFind: $\codevar{UF}$}

  \State $\beta=2$,
   $\rho_{\codevar{lo}}=0$
  \While{$|E_{\codevar{out}}| < (n-1)$}

  \State $\rho_{\codevar{hi}}=$ \Call{GetRho}{$R$, $\beta$}\label{line:memogetrho}
  \State $S_{l1}=$ \Call{GetPairs}{$R$, $\beta$, $\rho_{\codevar{lo}}$, $\rho_{\codevar{hi}}$, $\codevar{UF}$}\label{line:memogetpairs}
  \State $E_{l1}=$ \Call{GetEdges}{$S_{l1}$}  \Comment{Retrieves edges associated with pairs in $S_{l1}$}
  \State \Call{ParallelKruskal}{$E_{l1}$, $E_{\codevar{out}}$, $\codevar{UF}$}\label{line:memokruskal}
  \State $\beta=\beta\times 2$,
   $\rho_{\codevar{lo}}=\rho_{\codevar{hi}}$
  \EndWhile
  \EndProcedure
\end{algorithmic}
\end{algorithm}

The algorithm proceeds in rounds similar to parallel GeoFilterKruskal, and
maintains lower and upper bounds ($\rho_{\codevar{lo}}$ and
$\rho_{\codevar{hi}}$) on the weight of edges to be considered each
round.  On each round, it first computes $\rho_{\codevar{hi}}$ based on
$\beta$ by a single \kdt{} traversal, which will be elaborated below
(\cref{line:memogetrho}).  Then, together with $\rho_{\codevar{lo}}$
from the previous round ($\rho_{\codevar{lo}}=0$ on the first round),
the algorithm retrieves pairs with $\bccp$ distance in the range
$[\rho_{\codevar{lo}}, \rho_{\codevar{hi}})$
via a second \kdt{} traversal on \cref{line:memogetpairs}.
The edges corresponding to these pairs are then passed to
Kruskal's algorithm on \cref{line:memokruskal}.  An example of the
first round of the algorithm with MemoGFK is illustrated in
Figure~\ref{fig:wspd_ex}.  Without the optimization, the GFK algorithm
needs to first materialize all of the pairs in Round~1.  With MemoGFK,
$\rho_{\codevar{hi}}=d(e,Q_7)$ is computed via a tree traversal
on Line~\ref{line:memogetrho}, after which only the pairs in the set $S_{l1} = \{(a,d), (b,c), (f,g)\}$ are retrieved and materialized on Line~\ref{line:memogetpairs} via a second tree traversal.
Retrieving pairs only as needed reduces memory usage and improves performance.  The correctness of the algorithm follows from the fact that each round considers non-overlapping ranges of edge weights in increasing order until all edges are considered, or when MST is completed.

\input{figure-memofk}

Now we discuss the implementation details of the two-pass tree
traversal on Line~\ref{line:memogetrho}--\ref{line:memogetpairs}.  The
\textsc{GetRho} subroutine, which computes $\rho_{\codevar{hi}}$, does
so by finding the lower bound on the minimum separation of pairs whose
cardinality is greater than $\beta$ and are not yet connected in the
MST.
We traverse the \kdt starting at the root, in a similar way as when
computing the WSPD in \cref{alg:wspd}.
During the process, we update a global copy of $\rho_{\codevar{hi}}$
using \textsc{WriteMin} whenever we encounter a well-separated pair in
\textsc{FindPair}, with cardinality greater than $\beta$.  We can
prune the traversal once $|A|+|B|\leq\beta$, as all pairs that
originate from $(A,B)$ will have cardinality at most $\beta$.  We also
prune the traversal when the two children of a tree node are already
connected in the union-find structure, as these edges will not need to
be considered by Kruskal's algorithm.  In addition, we prune the
traversal when the distance between the bounding spheres of $A$ and
$B$, $d(A,B)$, is larger than $\rho_{\codevar{hi}}$, as its
descendants cannot produce a smaller distance. 

The \textsc{GetPairs} subroutine then retrieves all pairs whose points
are not yet connected in the union-find structure and have $\bccp$
distances in the range $[\rho_{\codevar{lo}}, \rho_{\codevar{hi}})$.
It does so also via a pruned traversal on the \kdt{}
starting from the root, similarly to \cref{alg:wspd}, but only
retrieves the useful pairs.  For a pair of nodes encountered in
the \textsc{FindPair} subroutine, we estimate the minimum and maximum
possible $\bccp$ between the pair using bounding sphere calculations,
an example of which is shown in Figure~\ref{fig:memogfk}a.  We prune
the traversal when $d_{\codevar{\max}}(A,B)<\rho_{\codevar{lo}}$, or
when $d(A,B)\geq \rho_{\codevar{hi}}$, 
in which case $\bccp(A,B)$
(as well as those of its recursive calls on descendant nodes) will be
outside of the range.  An example is shown in
Figure~\ref{fig:memogfk}b.  In addition, we also prune the traversal
if $A$ and $B$ are already connected in the MST, as an edge between
$A$ and $B$ will not be part of the MST.

We evaluate MemoGFK in Section~\ref{sec:experiment}.  We also use the memory
optimization for \hdbscan, which will be described
next.

%% file: wspd_ex.tex
\begin{figure}
  \vspace{-3pt}
  \centering
  \includegraphics[trim=0 0 0 0,clip,width = 0.7\columnwidth]{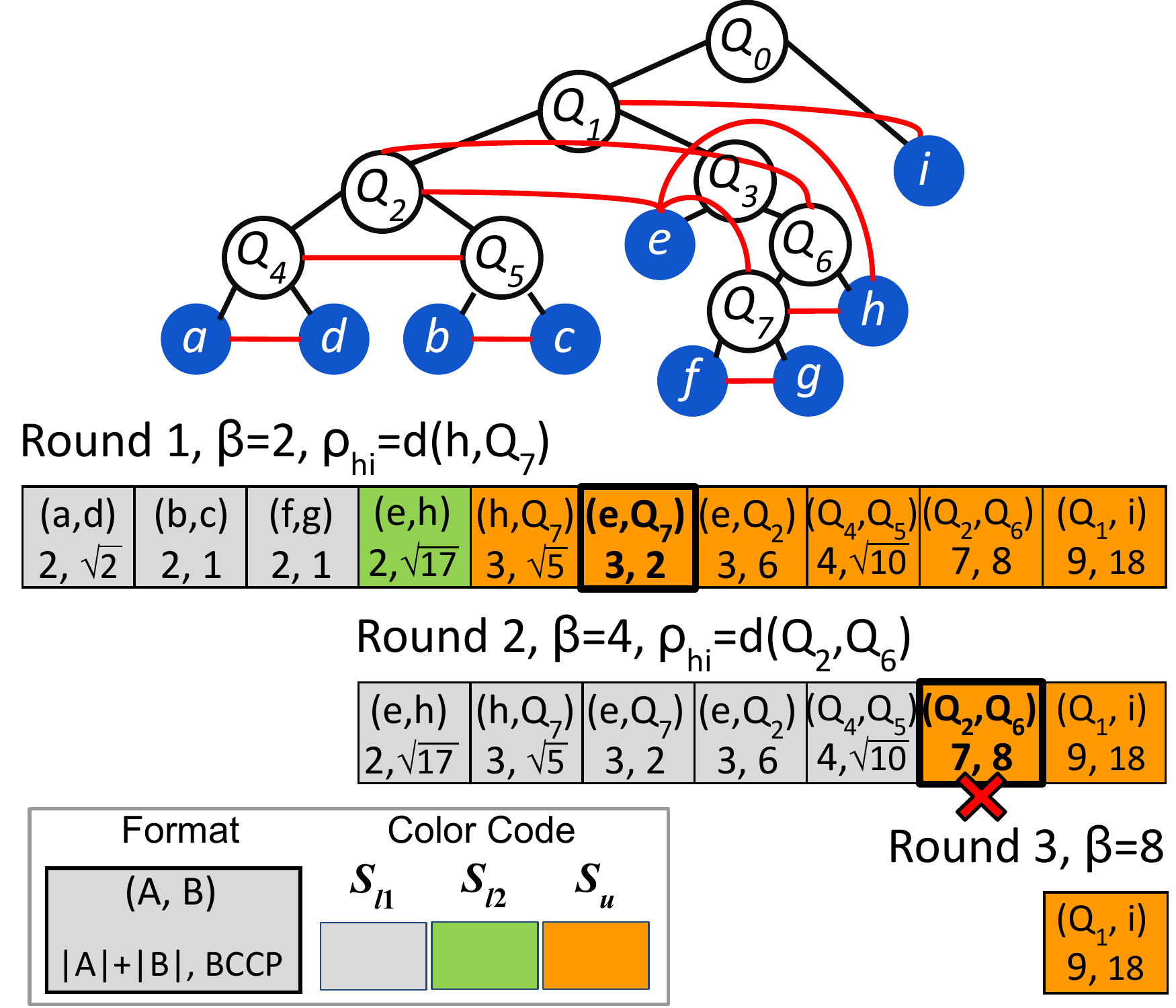}
  \caption{
    The is an example for both \textsc{GFK} (Algorithm~\ref{alg:gfk})
    and \textsc{MemoGFK} (Algorithm~\ref{alg:memogfk}) for EMST corresponding to the data set shown in Figure~\ref{fig:hdbscan}.
    The red lines linking tree nodes and the boxes drawn below represent
    well-separated pairs, 
    and the boxes also show the cardinality and \bccp{} value of the pair.
    Their correspondence with the symbols ($S_{l1}$, $S_{l2}$, $S_{u}$) in the pseudocode are color-coded.
    The pairs that generate $\rho_{\codevar{hi}}$ are bold-squared, and the pairs filtered out have a red cross.
    Using our MemoGFK optimization,
    only the pairs in $S_{l1}$ needs to be materialized, in contrast to needing to
    materialize all of the pairs in GFK.
  }
  \label{fig:wspd_ex}
\end{figure}

%% file: figure-memofk.tex
\begin{figure}[!t]
  \centering
  \vspace{1.5pt}
  \includegraphics[trim=0 0 0 45,width=0.7\linewidth]{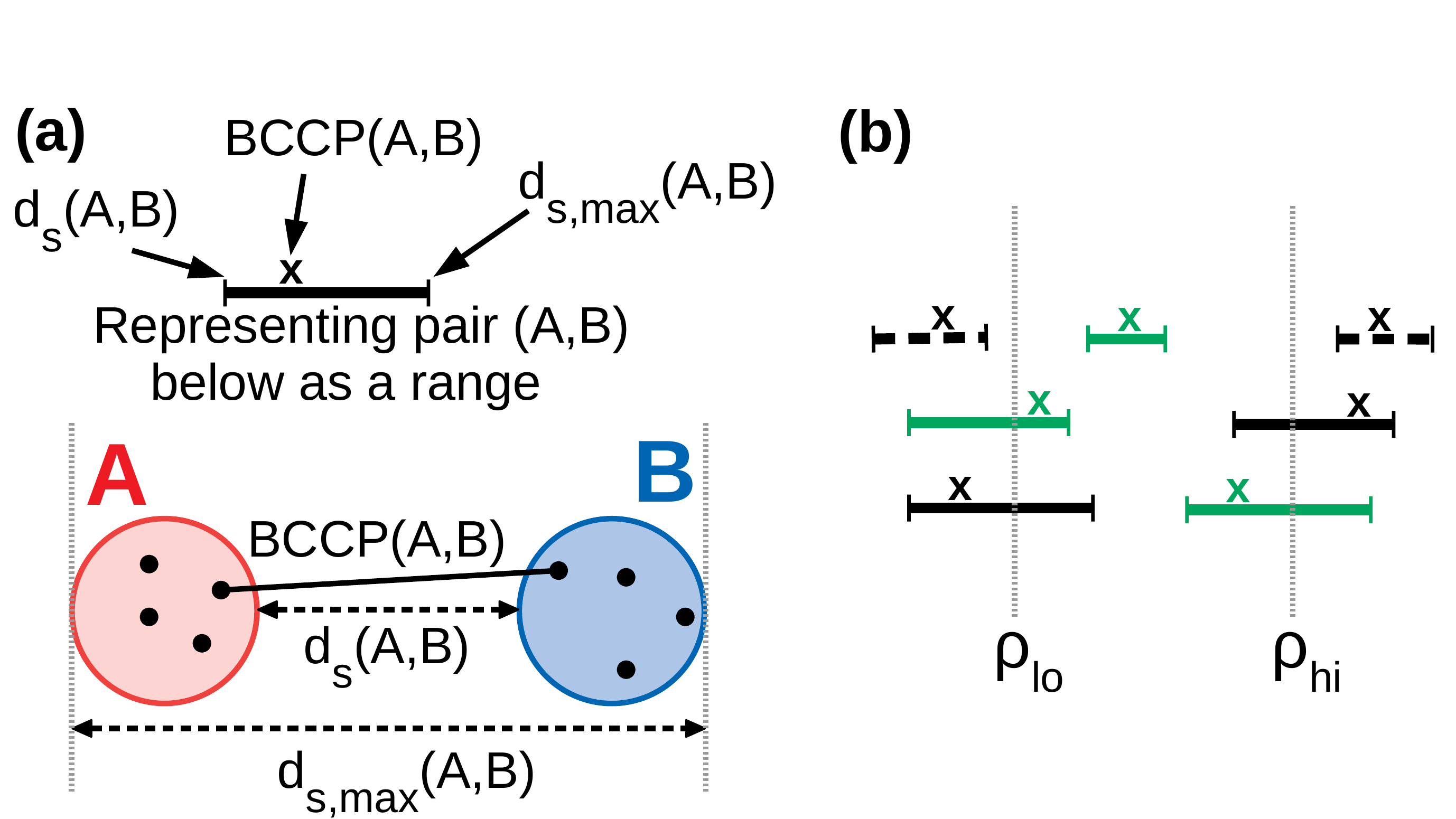}
  \caption{(a) shows a representation of a well-separated pair $(A,B)$
    as a line segment, based on the values of its
    $d(A,B)$ and $d_{\codevar{\max}}(A,B)$, which
    serve as the lower and upper bounds, respectively, for their $\bccp$ and the $\bccp$ of their
    descendants.
    The "x"'s on the line marks the value of the $\bccp$.
    (b) shows a conceptual example of tree node pairs
    encountered during a pruned tree traversal on Line~\ref{line:memogetpairs}
    of Algorithm~\ref{alg:memogfk},
    where the pairs are represented the same way as in (a).
    The pairs in solid green lines, if well-separated, will be retrieved and
    materialized because their $\bccp$s are within the
    $[\rho_{\codevar{lo}},\rho_{\codevar{hi}})$ range, whereas
    those in solid black lines will not as their $\bccp$s are
    out of range (although their $\bccp$s will still be computed, since their lower and upper bounds do not immediately put them out of range).
    The traversal will be pruned when encountering
    a pair represented by dotted lines as  their $\bccp$ and the $\bccp$ of their
    descendants will be out of range.
  } \label{fig:memogfk}
\end{figure}

%% file: hdbscan_dendro.tex
\section{Dendrogram and Reachability Plot}\label{sec:dendro}
We present a new parallel algorithm for generating a
dendrogram and reachability plot, given an unrooted tree with edge weights. Our
algorithm can be used for single-linkage clustering~\cite{gower1969minimum} by
passing the EMST as input, as well as for generating the \hdbscan
dendrogram and reachability plot (refer to Section~\ref{sec:prelims}
for definitions). In addition, our dendrogram
algorithm can be used in efficiently generating hierarchical
clusters using other linkage criteria (e.g.,~\cite{Yu2015a,Xu2001,Olson95}).

Sequentially, the dendrogram can be generated in a bottom-up
(agglomerative) fashion by sorting the edges by weight and processing
the edges in increasing order of
weight~\cite{mcinnes2017accelerated,BergGR17, mllner2011modern, gower1969minimum, hendrix2012parallel}.
Initially, all points are assigned their own clusters. Each edge
merges the clusters of its two endpoints, if they are in
different clusters, using a union-find data structure. The order of
the merges forms a tree structure, which is the dendrogram.
This takes $O(n\log n)$ work, but has little parallelism since the edges
need to be processed one at a time.  For \hdbscan,
we can generate the reachability plot directly from the input tree by
running Prim's algorithm on the tree edges starting from an arbitrary
vertex~\cite{Ankerst1999}. This approach takes $O(n\log n)$ work and is also
hard to parallelize efficiently, since Prim's algorithm is inherently sequential.

Our new parallel algorithm uses a top-down approach to generate the
dendrogram and reachability plot given a weighted tree. Our algorithm takes $O(n\log n)$ expected work and $O(\log^2n\log\log n)$ depth with high probability, and
hence is work-efficient. 

\subsection{Ordered Dendrogram}

We discuss the relationship between the dendrogram and \plot{}, which
are both used in \hdbscan.  It is known~\cite{sander2003automatic}
that a \plot{} can be converted into a dendrogram using a linear-work
algorithm for Cartesian tree construction~\cite{gabow1984scaling},
which can be parallelized~\cite{BlellochS14}.
However, converting in the other direction, which is what we need, is more challenging because
the children in dendrogram nodes are unordered, and can correspond to
many possible sequences, only one of which corresponds to the
traversal order in Prim's algorithm that defines the reachability
plot.

Therefore, for a specific starting point $s$, we define the
\defn{\newddg{}} of $s$, which is a dendrogram where its in-order traversal corresponds to the
\plot{} starting at point $s$.  With this definition, there is a
one-to-one correspondence between a \newddg{} and a \plot{}, and there
are a total of $n$ possible \newddg{}s and \plot{}s for an input of
size $n$.  Then, a \plot{} is just the in-order traversal of the leaves of an \newddg{}, and an \newddg{} is the corresponding Cartesian tree for the
\plot{}.

\subsection{A Novel Top-Down Algorithm}

We introduce a novel work-efficient parallel algorithm to compute
a dendrogram, which can be modified to compute an \newddg{} and its corresponding reachability plot.

\myparagraph{Warm-up} As a warm-up, we first propose a simple
top-down algorithm for constructing the dendrogram, which does not
quite give us the desired work and depth bounds.
We first generate an Euler tour on the input tree~\cite{JaJa92}. Then,
we delete the heaviest edge, which can be found in linear work and
$O(1)$ depth by checking all edges.  By definition, this edge will be
the root of the dendrogram, and removing this edge partitions the tree
into two subtrees corresponding to the two children of the root.  We
then convert our original Euler tour into two Euler tours, one for
each subtree, which can be done in constant work and depth by updating
a few pointers.  Next, we partition our list of edges into two lists,
one for each subproblem.
This can be done by applying list ranking on each Euler tour to
determine appropriate offsets for each edge in a new array associated
with its subproblem.
This step takes linear work and has $O(\log n)$ depth~\cite{JaJa92}.
Finally, we solve the two subproblems recursively.

Although the algorithm is simple, there is no guarantee that the
subproblems are of equal size. In the worst case, one of the
subproblems could contain all but one edges (e.g., if the tree is a
path with edge weights in increasing order), and the algorithm would
require $O(n)$ levels of recursion. The total work would then be
$O(n^2)$ and depth would be $O(n\log n)$, which is clearly
undesirable.

\myparagraph{An algorithm with $O(\log n)$ levels of recursion} We now describe a top-down approach that guarantees $O(\log n)$ levels of recursion. We define the \defn{heavy edges} of a tree with $n$ edges to be the $n/2$ (or any constant fraction of $n$) heaviest edges and the \defn{light edges} of a tree to be the remaining edges. Rather than using a single edge to partition the tree, we use the $n/2$ heaviest edges to partition the tree.
The heavy edges correspond to the part of the dendrogram closer to the root, which we refer to as the \emph{top} part of the dendrogram, and the light edges correspond to subtrees of the top part of the dendrogram. Therefore, we can recursively construct the dendrogram on the heavy edges and the dendrograms on the light edges in parallel.
Then, we insert the roots of the dendrograms for the light edges into the leaf nodes of the heavy-edge dendrogram. The base case is when there is a single edge, from which we can trivially generate a dendrogram.

An example is shown in Figure~\ref{fig:ddg}.
We first construct the Euler tour of the input tree
(Figure~\ref{fig:ddg}a). Then, we find the median edge based on edge weight, separate the heavy and light edges and compact them into a heavy-edge subproblem and multiple light-edge subproblems.
For the subproblems, we construct their Euler tours by adjusting pointers, and mark the position of each light-edge subproblem in the heavy-edge subproblem where it is detached.
Then, recursively and in parallel, we compute the dendrograms for each subproblem (Figure~\ref{fig:ddg}b).
After that, we insert the roots of the light-edge dendrograms to the appropriate leaf nodes in the heavy-edge dendrogram, as marked earlier (Figure~\ref{fig:ddg}c).

\input{figure-ddg}

\cref{fig:ddg} shows how this algorithm applies to the input in \cref{fig:hdbscan} with
 source vertex $a$.
The four heaviest edges $(b,c)$, $(d,e)$, $(f,h)$, and $(h,i)$ divide the tree into two light subproblems, consisting of $\{(a,d),(d,b)\}$ and $\{(e,g),(g,f)\}$.
The heavy edges form another subproblem. We mark vertices $b$ and $e$, where the light subproblems are detached.
After constructing the dendrogram for the three subproblems, we insert the light dendrograms at leaf nodes $b$ and $e$,
as shown in Figure~\ref{fig:ddg}b. It forms the correct dendrogram in Figure~\ref{fig:ddg}c.

We now describe the details of the steps to separate the subproblems and re-insert them into the final dendrogram. 

\myparagraph{Subproblem Finding}
To find the position in the heavy-edge dendrogram to insert a light-edge dendrogram at, every light-edge subproblem will be associated with a unique heavy edge.
The dendrogram of the light-edge subproblem will eventually connect to the corresponding leaf node in the heavy-edge dendrogram associated with it. We first explain how to separate the heavy-edge subproblem and the light-edge subproblems.

First, we compute the unweighted distance from every point to the starting point $s$ in the tree, and we refer to them as the \defn{vertex distances}. For the \newddg{}, $s$ is the starting point of the reachability plot, whereas $s$ can be an arbitrary vertex if the ordering property is not needed.
We compute the vertex distances by performing list ranking on the tree's Euler tour rooted at $s$.
These distances can be computed by labeling each downward edge (away from $s$) in the tree with a value of $1$ and each upward edge (towards $s$) in the tree with a value of $-1$, and running list ranking on the edges.
The vertex distances are computed only once.

We then identify the light-edge subproblems in parallel by using the vertex distances.
For each light edge $(u,v)$, we find an adjacent edge $(w,u)$ such that $w$ has smaller vertex distance than both $u$ and $v$. We call $(w,u)$ the \defn{predecessor edge} of $(u,v)$.
Each edge can only have one predecessor edge (an edge adjacent to $s$ will choose itself as the predecessor).
In a light-edge subproblem not containing the starting vertex $s$, the predecessor of each light edge will either be a light edge in the same light-edge subproblem, or a heavy edge. The edges in each light-edge subproblem will form a subtree based on the pointers to predecessor edges.
We can obtain the Euler tour of each light-edge subproblem by adjusting pointers of the original Euler tour. 
The next step is to run list ranking to propagate a unique label (the root's label of the subproblem subtree) of each light-edge subproblem to all edges in the same subproblem.
To create the Euler tour for the heavy subproblem, we contract the subtrees for the light-edge subproblems: for each light-edge subproblem, we map its leaves to its root using a parallel hash table. Now each heavy edge adjacent to a light-edge subproblem leaf can connect to the heavy edge adjacent to the light-edge subproblem root by looking it up in the hash table. The Euler tour for the heavy-edge subproblem can now be constructed by adjusting pointers.
We assign the label of the heavy-edge subproblem root to all of the heavy edges in parallel.
Then, we semisort the labeled edges to group edges of the same light-edge subproblems and the heavy-edge subproblem.
Finally, we recursively compute the dendrograms on the light-edge subproblems and the heavy-edge subproblem.
In the end, we connect the light-edge dendrogram for each subproblem to the heavy-edge dendrogram leaf node corresponding to the shared endpoint between the light-edge subproblem and its unique heavy predecessor edge.
For the light-edge subproblem containing the starting point $s$, we simply insert its light-edge dendrogram into the left-most leaf node of
the heavy-edge dendrogram.

Consider the example in Figure~\ref{fig:ddg}a.
The heavy-edge subproblem contains edges $\{(b,c), (d,e), (f,h), (h,i)\}$, and its dendrogram is shown in Figure~\ref{fig:ddg}b.
For the light-edge subproblem  $\{(e,g), (g,f)\}$, $(e,g)$ has heavy predecessor edge $(d,e)$, and $(g,f)$ has light predecessor edge $(e,g)$.
The unique heavy edge associated with the light-edge subproblem is hence $(d,e)$, with which it shares vertex $e$.
Hence, we insert the light-edge dendrogram for the subproblem into leaf node $e$ in the heavy-edge dendrogram, as shown in Figure~\ref{fig:ddg}b.
The light-edge subproblem containing $\{(a,d),(d,b)\}$ contains the starting point $s=a$, and so we insert its dendrogram into the leftmost leaf node $b$ of the heavy-edge dendrogram, as shown in Figure~\ref{fig:ddg}b.

We first show that our algorithm correctly computes a dendrogram, and analyze its cost bounds (Theorem~\ref{thm:dendro}). Then, we describe and analyze additional steps needed to generate an ordered dendrogram and obtain a reachability plot from it (Theorem~\ref{thm:dendro-ordering}). 

\begin{theorem}\label{thm:dendro}
Given a weighted spanning tree with $n$ vertices, we can compute a dendrogram in $O(n\log n)$ expected work and $O(\log^2 n\log \log n)$ depth with high probability.
\end{theorem}

\begin{proof}
We first prove that our algorithm correctly produces a dendrogram.
In the base case, we have one edge $(u,v)$, and the algorithm produces a tree with a root representing $(u,v)$, and with $u$ and $v$ as children of the root, which is trivially a dendrogram.
We now inductively hypothesize that recursive calls to our algorithm correctly produce 
dendrograms.
The heavy subproblem recursively computes a top dendrogram consisting of all of the heavy edges, and the light subproblems form dendrograms consisting of light edges. We replace the leaf vertices in the top dendrogram associated with light subproblems by the roots of the  dendrograms on light edges.
Since the edges in the heavy subproblem are heavier than all edges in light subproblems, and are also ancestors of the light edges in the resulting tree, this gives a valid dendrogram.

We now analyze the cost of the algorithm.
To generate the Euler tour at the beginning, we first sort the edges and create an adjacency list representation, which takes $O(n\log n)$ work and $O(\log n)$ depth~\cite{Cole88}.
Next, we root the tree, which can be done by list ranking on the Euler tour of the tree.
Then, we compute the vertex distances to $s$ using another round of list ranking based on the rooted tree.

There are $O(\log n)$ recursive levels since the subproblem sizes are at most half of the original problem.
We now show that each recursive level takes linear expected work and polylogarithmic depth with high probability.
Note that we cannot afford to sort the edges on every recursive level, since that would take $O(n\log n)$ work per level.
However, we only need to know which edges are heavy and which are light, and so we can use parallel selection~\cite{JaJa92} to find the median and partition the edges into two sets. This takes $O(n)$ work and $O(\log n\log \log n)$ depth.
Identifying predecessor edges takes a total of $O(n)$ work and $O(1)$ depth:
first, we find and record for each vertex its edge where the other endpoint has a smaller vertex distance than it (using \textsc{WriteMin}); then, the predecessor of each edge can be found by checking the recorded edge for its endpoint with smaller vertex distance. 
We then use list ranking to assign labels to each subproblem, which takes $O(n)$ work and
$O(\log n)$ depth~\cite{JaJa92}. 
The hash table operations to contract and look up the light-edge subproblems
contribute $O(n)$  work and $O(\log n)$ depth with high probability.
The semisort to group the subproblems takes $O(n)$ expected work and $O(\log n)$ depth with high probability.
Attaching the light-edge dendrograms into the heavy-edge dendrogram takes $O(n)$ work and $O(1)$ depth across all subproblems.
Multiplying the bounds by the number of levels of recursion proves the theorem.
\end{proof}

\vspace{-3pt}
\begin{theorem}\label{thm:dendro-ordering}
Given a starting vertex $s$, 
we can generate an \newddg{} and reachability plot in the same cost bounds as in Theorem~\ref{thm:dendro}.
\end{theorem}
\vspace{-3pt}

\begin{proof}

We have computed the vertex distances of all vertices from $s$. 
When generating the \newddg{} and constructing each internal node of the dendrogram corresponding to an edge $(u,v)$, and without loss of generality let $u$ have a smaller vertex distance than $v$, our algorithm puts the result of the subproblem attached to $u$ in the left subtree, and that of $v$ in the right subtree. This additional comparison does not increase the work and depth of our algorithm.

Our algorithm recursively builds \newddg{}s on the heavy-edge subproblem and on each of the light-edge subproblems, which we assume to be correct by induction. The base case is a single edge $(u,v)$, and without loss of generality let $u$ have a smaller vertex distance than $v$. Then, the dendrogram will contain a root node representing edge $(u,v)$, with
$u$ as its left child and $v$ as its right child. Prim's algorithm would visit $u$ before $v$, and the in-order traversal of the dendrogram does as well, so this is an \newddg{}. 

We now argue that the way that light-edge dendrograms are attached to the leaves of the heavy-edge dendrogram correctly produces an \newddg{}. First, consider a light-edge subproblem that contains the source vertex $s$. In this case, 
its dendrogram is attached as the leftmost leaf of the heavy-edge dendrogram, and will be the first to be traversed in the in-order traversal. The vertices in the light-edge subproblem form a connected component $A$. They will be traversed before any other vertices in Prim's algorithm because all incident edges that leave $A$ are heavy edges, and thus are heavier than any edge in $A$. Therefore, vertices outside of $A$ can only be visited after all vertices in $A$ have been visited, which correctly corresponds to the in-order traversal.

Next, we consider the case where the light-edge subproblem does not contain $s$.
Let $(u,v)$ be the predecessor edge of the light-edge subproblem, 
and let $A$ be the component containing the edges in the light-edge subproblem ($v$ is a vertex in $A$).
Now, consider a different light-edge subproblem that does not contain $s$, whose predecessor edge is $(x,y)$, and let $B$ be the component containing the edges in this subproblem ($y$ is a vertex in $B$).
By construction, we know that $A$ is in the right subtree of the dendrogram node corresponding to edge $(u,v)$ and $B$ is in the right subtree of node corresponding to $(x,y)$.
The ordering between $A$ and $B$ is correct as long as they are on different sides of either node $(u,v)$ or node $(x,y)$.
For example, if $B$ is in the left subtree of node $(u,v)$, then its vertices appear before $A$ in the in-order traversal of the dendrogram. 
By the inductive hypothesis on the heavy-edge subproblem, in Prim's order, $B$ will be traversed before $(u,v)$, and $(u,v)$ is traversed before $A$. 
We can apply a similar argument to all other cases where $A$ and $B$ are on different sides of either node $(u,v)$ or node $(x,y)$.

We are concerned with the case where $A$ and $B$ are both in the right subtrees of the nodes representing their predecessor edges. We prove by contradiction that this cannot happen.
Without loss of generality, suppose node $(x,y)$ is in the right subtree of node $(u,v)$, and let both $A$ and $B$ be in the right subtree of $(x,y)$.
There exists a lowest common ancestor (LCA) node $(x',y')$ of $A$ and $B$. $(x',y')$ must be a heavy edge in the right subtree of $(x,y)$.
By properties of the LCA, $A$ and $B$ are in different subtrees of node $(x',y')$. Without loss of generality, let $A$ be in the left subtree. 
Now consider edge $(x',y')$ in the tree. By the inductive hypothesis on the heavy-edge dendrogram, in Prim's traversal order, we must first visit the leaf that $A$ attaches  to (and hence $A$) before visiting $(x',y')$, which must be visited before the leaf that $B$ attaches to (and hence $B$).
On the other hand, edge $(x,y)$ is also along the same path since it is the predecessor of $B$.
Thus, we must either have $(x',y')$ in $(x,y)$'s left subtree
or $(x,y)$ in $(x',y')$'s right subtree,
which is a contradiction to $(x',y')$ being in the right subtree of $(x,y)$.

We have shown that given any two light-edge subproblems, their relative ordering after being attached to the heavy-edge dendrogram is correct.
Since the heavy-edge dendrogram is an \newddg{} by induction, the order in which the light-edge subproblems are traversed is correct. Furthermore, each light-edge subproblem generates an \newddg{} by induction. Therefore, the overall dendrogram is an \newddg{}.

Once the \newddg{} is computed, we can use list ranking to
perform an in-order traversal on the Euler tour of the dendrogram
to give each node a rank, and write them out in order.
We then filter out the non-leaf nodes to
to obtain the \plot{}.
Both list ranking and filtering take $O(n)$ work and $O(\log n)$ depth.
\end{proof}

\myparagraph{Implementation}
In our implementation, we simplify the process of finding the subproblems
by using a sequential procedure rather than performing parallel list ranking,
because in most cases parallelizing over the different subproblems already
provides sufficient parallelism.
We set the number of heavy edges to $n/10$, which we found to give better performance in practice, and also preserves the theoretical bounds.
We switch to the sequential dendrogram construction algorithm when the problem size falls below $n/2$.

%% file: figure-ddg.tex
\begin{figure}
    \vspace{-12pt}
    \includegraphics[trim={0 0 65 100},clip,width = 0.82\columnwidth]{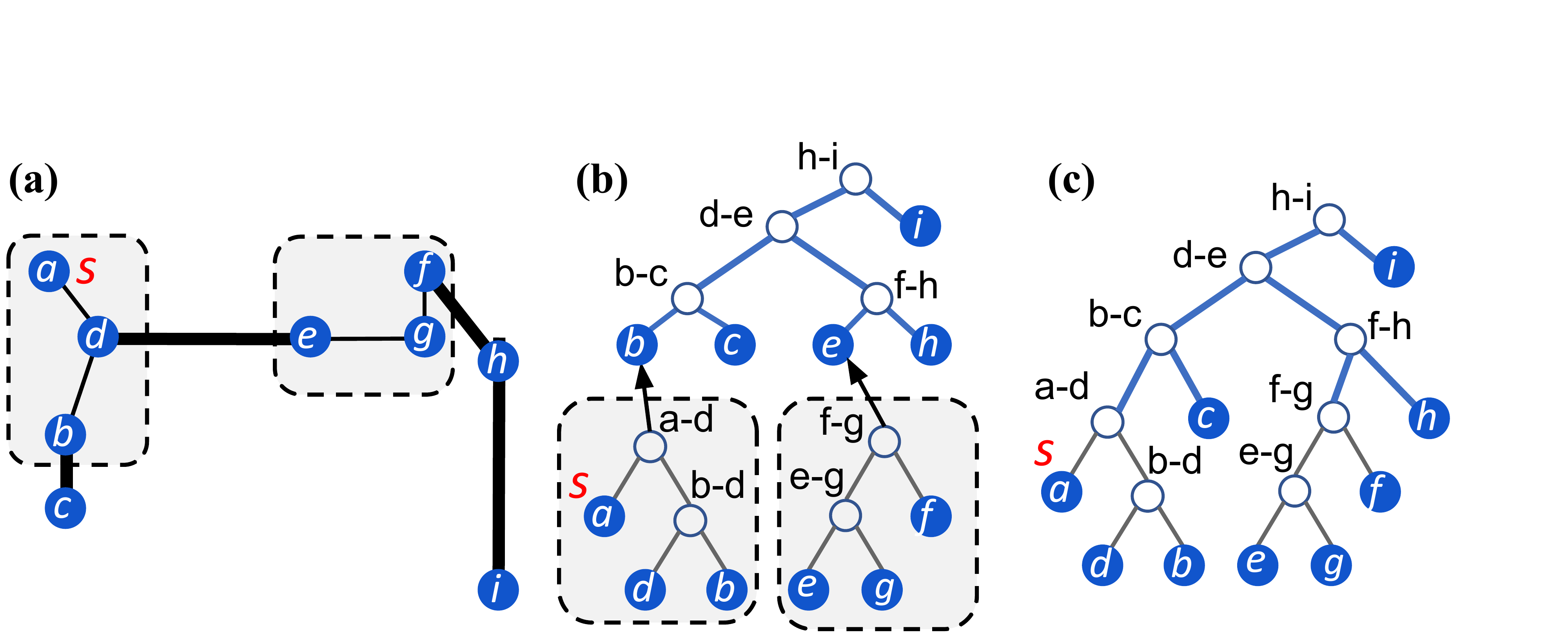}
    \caption{
      An example of the dendrogram construction algorithm on the tree from \cref{fig:hdbscan}. 
      The input tree is shown in
      (a). The 4 heavy edges are in bold. 
      We have three subproblems---one for the
      heavy edges and two for the light edges.  The dendrograms for the
      subproblems are generated recursively, as shown in (b).  The
      edge labeled on an internal node is the edge whose removal
      splits a cluster into the two clusters represented by its
      children.  As shown in (c), we insert the roots of the
      dendrograms for the light edges at  the corresponding
      leaf nodes of the heavy-edge dendrogram.  For the ordered
      dendrogram, the in-order traversal of the leaves 
      corresponds to the \plot{} shown in \cref{fig:hdbscan} when the starting point $s=a$.
    }     \vspace{-5pt}
    \label{fig:ddg}
\end{figure}

%% file: experiment.tex
\section{Experiments}\label{sec:experiment}

\myparagraph{Environment} We perform experiments on
an Amazon EC2 instance with 2
$\times$ Intel Xeon Platinum 8275CL (3.00GHz) CPUs for a total of 48
cores with two-way hyper-threading, and 192 GB of RAM.  By default, we use all cores
with hyper-threading.  We use the \texttt{g++} compiler (version 7.4)
with the \texttt{-O3} flag, and use Cilk for parallelism~\cite{cilkplus}.  We
do not report
times for tests that exceed 3 hours.

We test the following implementations for EMST (note that the EMST problem does not include dendrogram generation):
\begin{itemize}[topsep=1pt,itemsep=0pt,parsep=0pt,leftmargin=10pt]
\item \defn{EMST-Naive}: The method of creating a graph
  with the $\bccp$ edges from all well-separated pairs and then
  running MST on it.

\item \defn{EMST-GFK}: The parallel GeoFilterKruskal algorithm described
  in Section~\ref{sec:gfk} (Algorithm~\ref{alg:gfk}).
\item \defn{EMST-MemoGFK}: The parallel GeoFilterKruskal algorithm
  with the memory optimization
  described in Section~\ref{sec:impl:memogfk} (Algorithm~\ref{alg:memogfk}).
\iffullversion
\item \defn{EMST-Delaunay}: The method of computing an MST on a Delaunay triangulation for 2D data sets described in \cref{sec:emst_2d}.
\fi
\end{itemize}

We test the following implementations for \hdbscan:
\begin{itemize}[topsep=1pt,itemsep=0pt,parsep=0pt,leftmargin=10pt]
\item \defn{\hdbscan-GanTao}: The modified algorithm of Gan and Tao for exact \hdbscan described in Section~\ref{sec:hdbscan_wspd:existing}.
\item \defn{\hdbscan-MemoGFK}: The \hdbscan algorithm using our new definition of well-separation described in Section~\ref{sec:hdbscan_wspd:our}.
\end{itemize}

Both  \defn{\hdbscan-GanTao} and \defn{\hdbscan-MemoGFK}  use the memory optimization described in Section~\ref{sec:impl:memogfk}.  All \hdbscan running times include constructing an MST of the mutual reachability graph and computing the ordered dendrogram.  We use a default value of $\minpts=10$ (unless specified otherwise), which is also adopted in previous work~\cite{Campello2015, mcinnes2017accelerated, Gan2018}.

Our algorithms are designed for multicores, as we found that
multicores are able to process the largest data sets in the literature
for these problems (machines with several terabytes of RAM can be
rented at reasonable costs on the cloud).  Our multicore
implementations achieve significant speedups over existing
implementations in both the multicore and distributed memory
contexts.

\myparagraph{Data Sets} We use the synthetic seed spreader data sets
produced by the generator in~\cite{GanT17}. It produces points
generated by a random walk in a local neighborhood (\defn{SS-varden}).
We also use \defn{UniformFill} that contains
points distributed uniformly at random inside a bounding hypergrid
with side length $\sqrt{n}$ where $n$ is the total number of points.
We generated the synthetic data sets with 10 million points (unless
specified otherwise) for dimensions $d=2,3,5,7$.

We use the following real-world data sets.
\defn{GeoLife}~\cite{Zheng2008, GeoLifeURL}
is a 3-dimensional data set with $24,876,978$ data points. This data set
contains user location data (longitude, latitude, and altitude), and
is extremely skewed.  \defn{Household}~\cite{UCI, HouseholdURL}
is a 7-dimensional data set with $2,049,280$ points representing electricity
consumption measurements in households.
\defn{HT}~\cite{Huerta2016OnlineHA, HTURL}
is a 10-dimensional data set with $928,991$ data points containing home
sensor data.  \defn{CHEM}~\cite{fonollosa2015reservoir, CHEMURL}
is a 16-dimensional data set with $4,208,261$ data points containing chemical sensor data.
All of the data sets fit in the RAM of our machine.

\myparagraph{Comparison with Previous Implementations}
For EMST, we tested the sequential Dual-Tree Boruvka 
algorithm of March
et al.~\cite{march2010fast} (part of \texttt{mlpack}), and
our single-threaded \textsc{EMST-MemoGFK} times are 0.89--4.17 (2.44 on average) times faster.
\iffullversion
Raw running times for \texttt{mlpack} are presented in Table~\ref{table:emst_mlpack}.
\fi
We also tested McInnes and Healy's sequential \hdbscan implementation
which is based on Dual-Tree Boruvka~\cite{mcinnes2017accelerated}.  We were unable
to run their code on our data sets with 10 million points in a
reasonable amount of time.  On a smaller data set with 1 million points
(2D-SS-varden-1M), their code takes around 90 seconds to compute the MST
and dendrogram, which is 10 times slower than our
\hdbscan-MemoGFK implementation on a single thread, due to their code using Python and having fewer optimizations.
We observed a similar trend on other data sets for McInnes and Healy's implementation.

The GFK algorithm implementation for EMST of~\cite{ChatterjeeCK10} in the
\texttt{Stann} library supports multicore execution using
OpenMP.  We found that, in parallel, their 
GFK implementation always runs much slower when using all 48
cores than running sequentially, and so we decided to not include their
parallel running times in our experiments.  In addition, our own
sequential implementation of the same algorithm is 0.79--2.43x
(1.23x on average) faster than theirs, and so we parallelize our own
version as a baseline.
We also tested the multicore implementation of the parallel OPTICS
algorithm in~\cite{Patwary2013} using all 48 cores on our machine.
Their code exceeded our 3-hour time limit for our
data sets with 10 million points.  On a smaller data set of 1 million
points (2D-SS-varden-1M), their code took 7988.52 seconds, whereas our
fastest parallel implementations take only a few seconds.
We also compared with the parallel \hdbscan code by Santos et al.~\cite{santos2019hdbscanmapreduce}, which mainly focuses on
approximate \hdbscan in distributed memory. As reported in their paper, for the HT data set with $\minpts = 30$, their code on 60 cores
takes 42.54 and 31450.89 minutes to build the approximate and exact
MST, respectively, and 124.82 minutes to build the dendrogram. In
contrast, our fastest implementation using 48 cores builds the MST in
under 3 seconds, and the dendrogram in under a second.

Overall, we found the fastest sequential methods for EMST and \hdbscan to be our EMST-MemoGFK and HDBSCAN*-MemoGFK methods running on 1 thread. Therefore, we also based our parallel implementations on these methods.

\input{results_emst}
\input{results_hdbscan}
\iffullversion
\input{results_bar_4x}
\else
\input{results_bar}
\fi

\myparagraph{Performance of Our Implementations}
\iffullversion
Raw running times for our implementations are presented in Tables~\ref{table:emst}
and~\ref{table:hdbscan} in the Appendix.
\fi
Table~\ref{tab:speedup} shows the self-relative speedups and speedups
over the fastest sequential time of our parallel implementations on 48 cores.  Figures~\ref{plot:emst_speedup} and~\ref{plot:hdbscan_speedup} show the parallel speedup as a function of
thread count for our implementations of EMST and \hdbscan with
$\minpts=10$, respectively, against the fastest sequential times.  For most data sets, we see additional speedups from using hyper-threading compared to just using a single thread per core.
A decomposition of parallel timings for all of our implementations on
several data sets is presented in Figure~\ref{plot:bar_plots}.

\input{speedup_table}

\myparagraph{EMST Results} In Figure~\ref{plot:emst_speedup}, we see
that our fastest EMST implementations (EMST-MemoGFK) achieve
good speedups over the best sequential times, ranging from 14.61--55.89x on 48 cores with hyper-threading.
On the lower end, 10D-HT-0.93M has a speedup of 14.61x (Figure~\ref{plot:emst_speedup}k).
This is because for a small data set, the total work done is small and the parallelization overhead becomes prominent.

EMST-MemoGFK significantly outperforms
EMST-GFK and EMST-Naive by up to 17.69x and 8.63x, respectively, due
to its memory optimization, which reduces memory traffic.  We note that
EMST-GFK does not get good speedup, and is slower than EMST-Naive in
all subplots of Figure~\ref{plot:emst_speedup}.
This is because the WSPD input to
EMST-GFK ($S$ in Algorithm~\ref{alg:gfk}) needs to store references to the well-separated pair as well as the $\bccp$ points and distances,
whereas EMST-Naive only needs to store the $\bccp$ points and distances.
This leads to increased memory traffic for EMST-GFK for operations on $S$ and its subarrays, which outweighs its advantage of computing fewer $\bccp$s.
This is evident from Figure~\ref{plot:bar_plots}, which shows that EMST-GFK spends
more time in WSPD, but less time in Kruskal compared to
EMST-Naive. EMST-MemoGFK spends the least amount of time in WSPD due to its
pruning optimizations, while spending a similar amount of
time in Kruskal as EMST-GFK.
\iffullversion
Finally, the EMST-Delaunay implementation performs reasonably well, being only slightly (1.22--1.57x) slower than EMST-MemoGFK; however, it is only applicable for 2D data sets.
\fi

\myparagraph{\hdbscan Results} In
Figure~\ref{plot:hdbscan_speedup}, we see that our \hdbscan-MemoGFK
method achieves good speedups over the best sequential times, ranging from 11.13--46.69x on 48 cores.
Similar to EMST, we observe a similar lower speedup for 10D-HT-0.93M due to its small size, and observe higher speedups for larger data sets. The dendrogram construction takes at least 50\% of the total time for Figures 7a, b, and e--h, and hence has a large impact on the overall scalability. We discuss the dendrogram scalability separately.

We find that \hdbscan-MemoGFK consistently
outperforms \hdbscan-GanTao due to having a fewer number of
well-separated pairs (2.5--10.29x fewer) using the new definition of
well-separation.  This is also evident in Figure~\ref{plot:bar_plots},
where we see that \hdbscan-MemoGFK spends much less time than
\hdbscan-GanTao in WSPD computation.

We tried varying \minpts{} over a range from 10 to 50 for our \hdbscan
implementations and found just a moderate increase in the running time for increasing \minpts.

\myparagraph{MemoGFK Memory Usage}
Overall, the MemoGFK method for both EMST and \hdbscan reduces memory
usage by up to 10x compared to materializing all WSPD pairs in a naive
implementation.

\myparagraph{Dendrogram Results}
We separately report the performance of our parallel dendrogram algorithm in Figure~\ref{plot:dendro-bar}, which  shows the speedups and running times on all of our data sets. We see that the
 parallel speedup ranges from
5.69--49.74x (with an average of 17.93x) for the \hdbscan MST with \minpts=10, and 5.35--52.58x (with an average 20.64x) for single-linkage clustering, which is solved by generating a dendrogram on the EMST.
Dendrogram construction for single-linkage clustering shows higher scalability because the heavy edges are more uniformly
distributed in space, which creates a larger number of light-edge subproblems and increases parallelism.
In contrast, for \hdbscan, which has a higher value of $\minpts$, the sparse regions in the space tend to have clusters of edges with large weights even if some of them have small Euclidean distances,
since these edges have high mutual reachability distances.
Therefore, these heavy edges are less likely to divide up the edges into a uniform distribution of subproblems in the space, leading to lower parallelism.
On the other hand, we observe that across all data sets, the dendrogram for single-linkage clustering takes
an average of 16.44 seconds, whereas the dendrogram for \hdbscan takes an average of 9.27 seconds.
This is because the single-linkage clustering generates more light-edge subproblems and hence requires more work.
While it is possible to tune the fraction of heavy edges for different values of \minpts, we found that using $n/10$ heavy edges works reasonably well in all cases.

\input{results_dendrobar}




%% file: results_emst.tex
\iffullversion

\begin{figure*}
\begin{center}
\vspace{-1em}
\includegraphics[width=\textwidth]{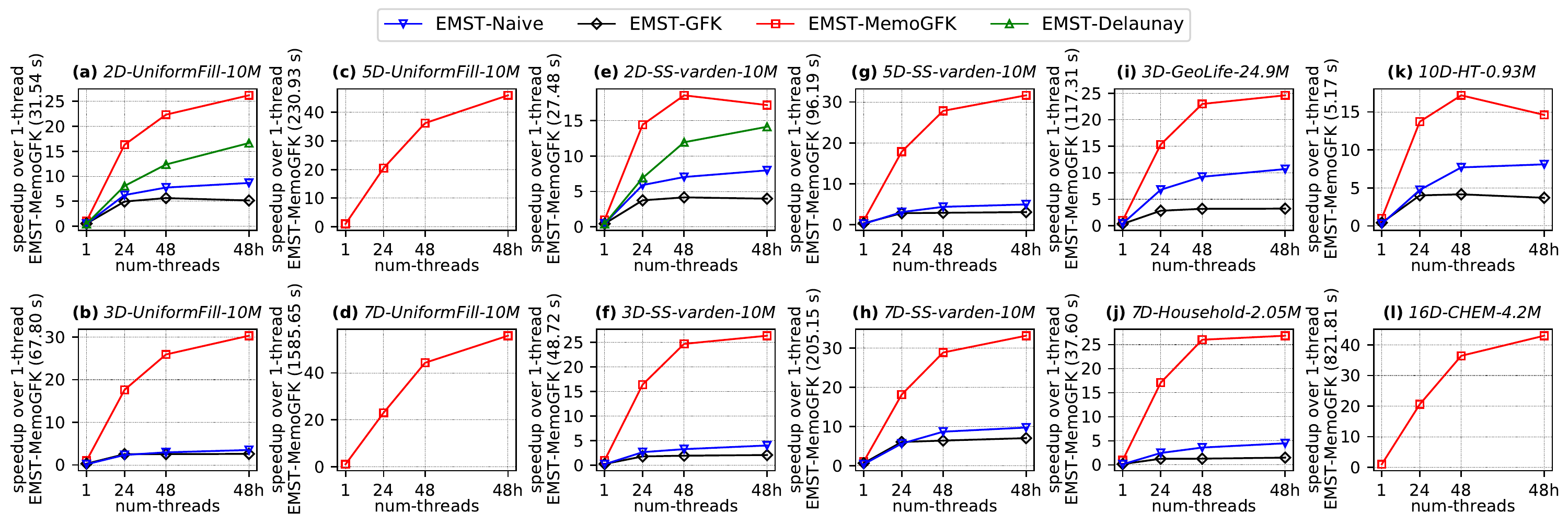}
\caption{Speedup of EMST implementations over the \emph{best} serial baselines vs.
  thread count. The best serial baseline and its running time for each data set
  is shown on the $y$-axis label. ``48h''
  on the $x$-axis refers to 48 cores with hyper-threading.}
\label{plot:emst_speedup}
\end{center}
\end{figure*}

\else

\begin{figure*}
\begin{center}
\vspace{-5pt}
\includegraphics[width=0.95\textwidth]{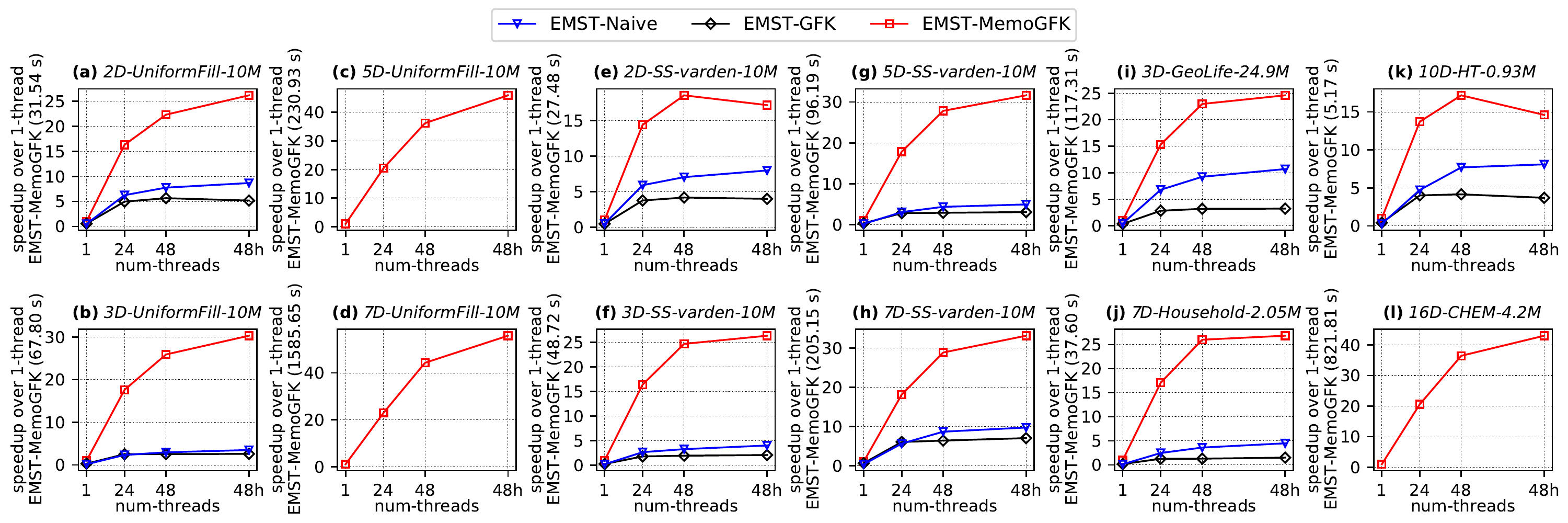}
\caption{Speedup of EMST implementations over the \emph{best} serial baselines vs.
  thread count. The best serial baseline and its running time for each data set
  is shown on the $y$-axis label. ``48h''
  on the $x$-axis refers to 48 cores with hyper-threading. }
\label{plot:emst_speedup}
\end{center}
\end{figure*}

\fi

%% file: results_hdbscan.tex
\begin{figure*}
\vspace{-5pt}
\begin{center}
\includegraphics[width=0.95\textwidth,height=2.3in]{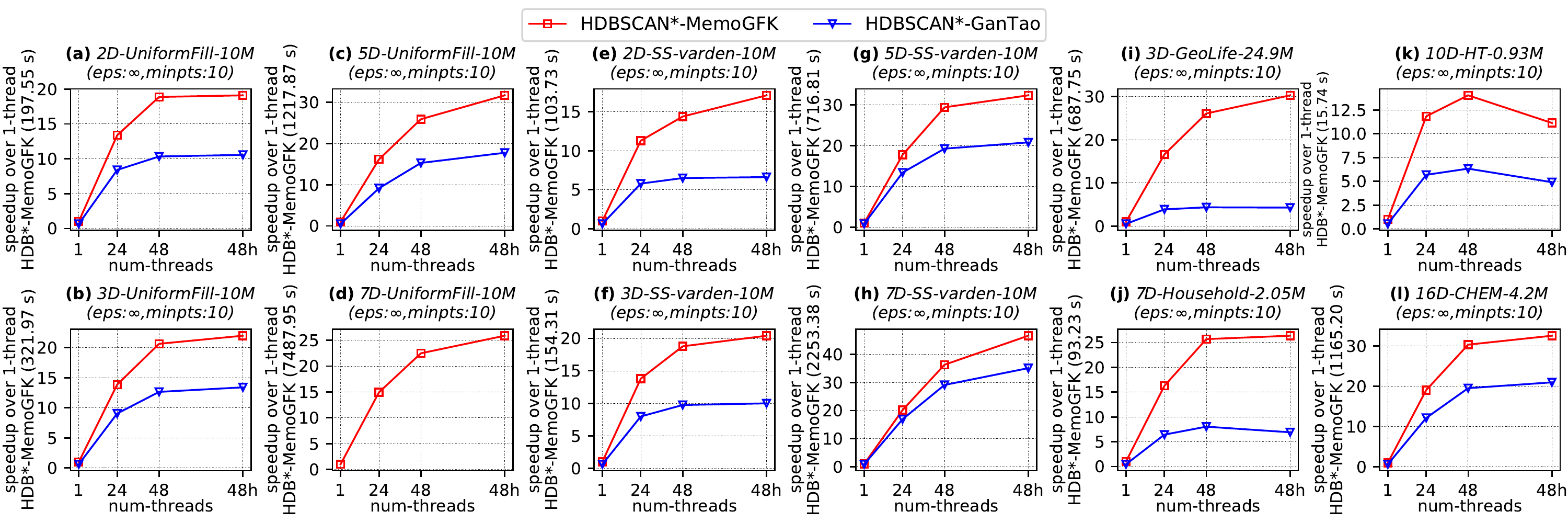}
\caption{Speedup of implementations for \hdbscan MST generation over the \emph{best} serial baselines vs.
  thread count using
   $\minpts=10$.
  The best serial baseline and its running time for each data set
  is shown on the $y$-axis label. ``48h''
  on the $x$-axis refers to 48 cores with hyper-threading.}
\label{plot:hdbscan_speedup}
\end{center}
\end{figure*}

%% file: results_bar_4x.tex
\begin{figure*}
  \begin{center}
    \begin{subfigure}[b]{3in}
      \center
      \includegraphics[trim=0 0 100 0,clip,width=0.8\textwidth]{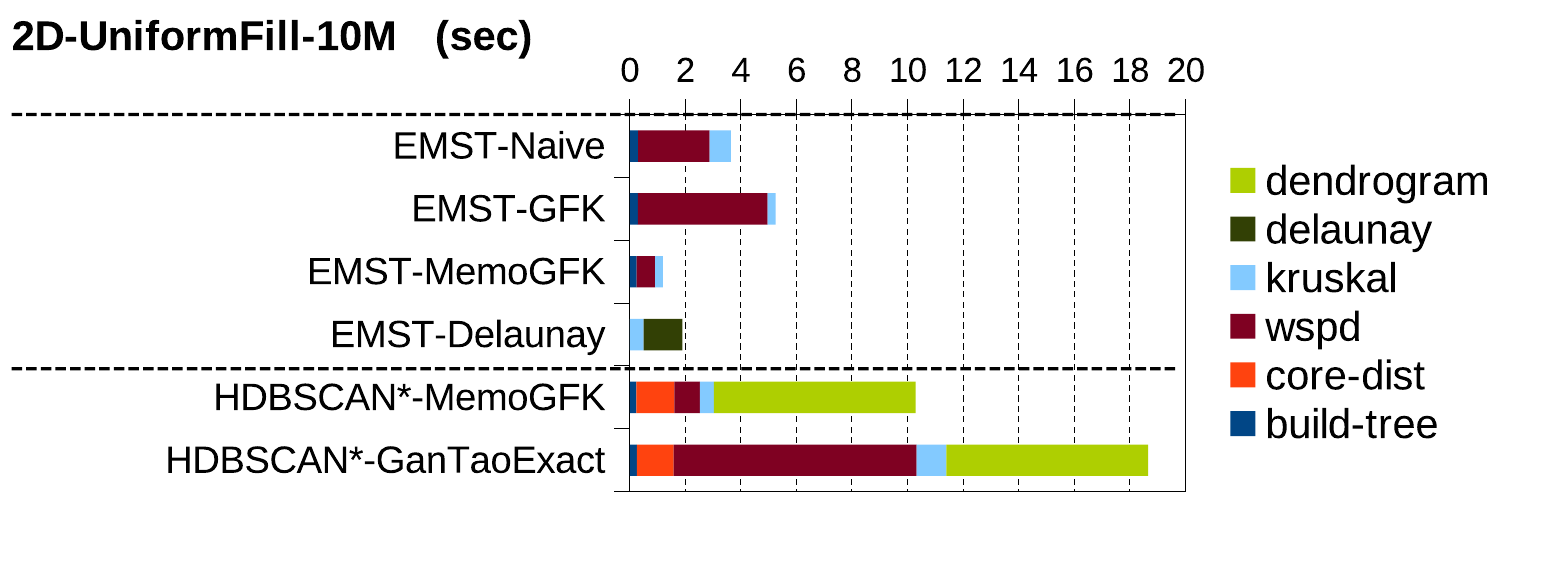}
    \end{subfigure} %
    \begin{subfigure}[b]{3in}
      \center
      \includegraphics[trim=0 0 0 0,clip,width=1.05\textwidth]{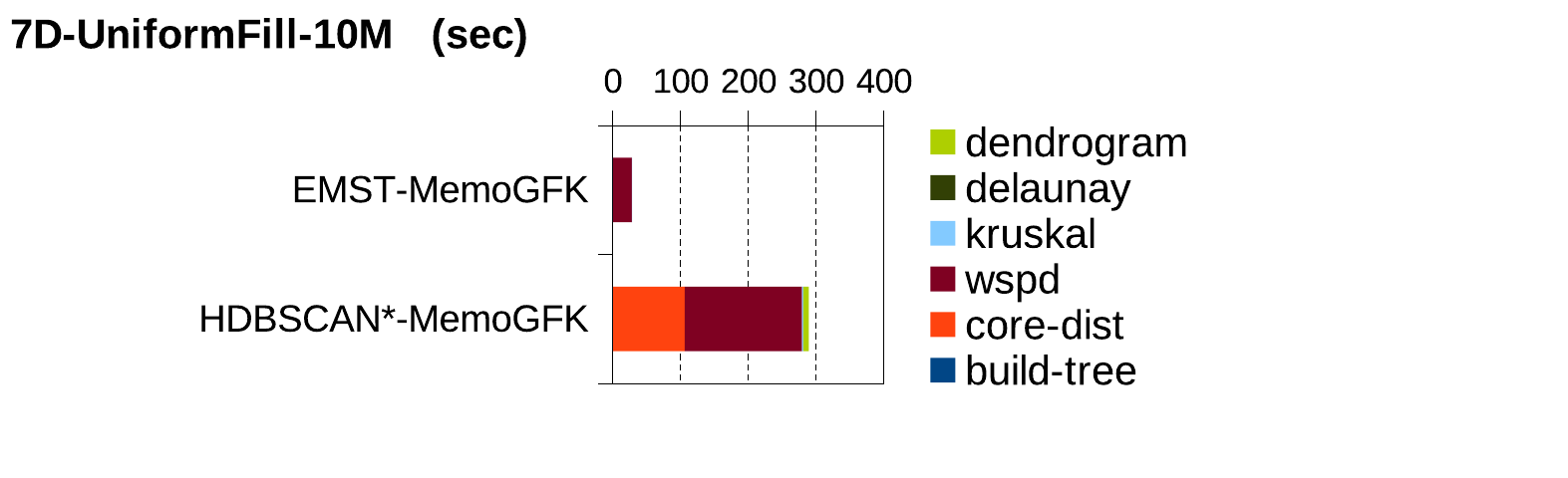}
    \end{subfigure} %
    \begin{subfigure}[b]{3in}
      \center
      \includegraphics[trim=0 0 110 0,clip,width=0.8\textwidth]{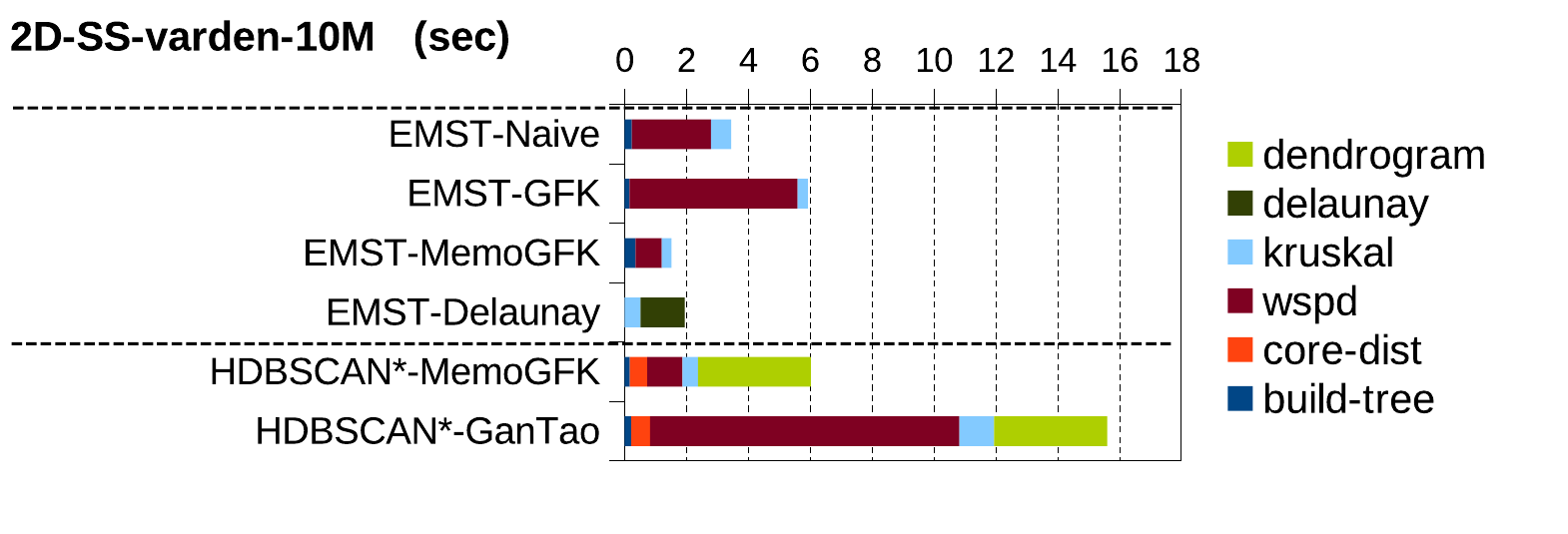}
    \end{subfigure} %
    \begin{subfigure}[b]{3in}
      \center
      \includegraphics[trim=0 0 106 0,clip,width=0.82\textwidth]{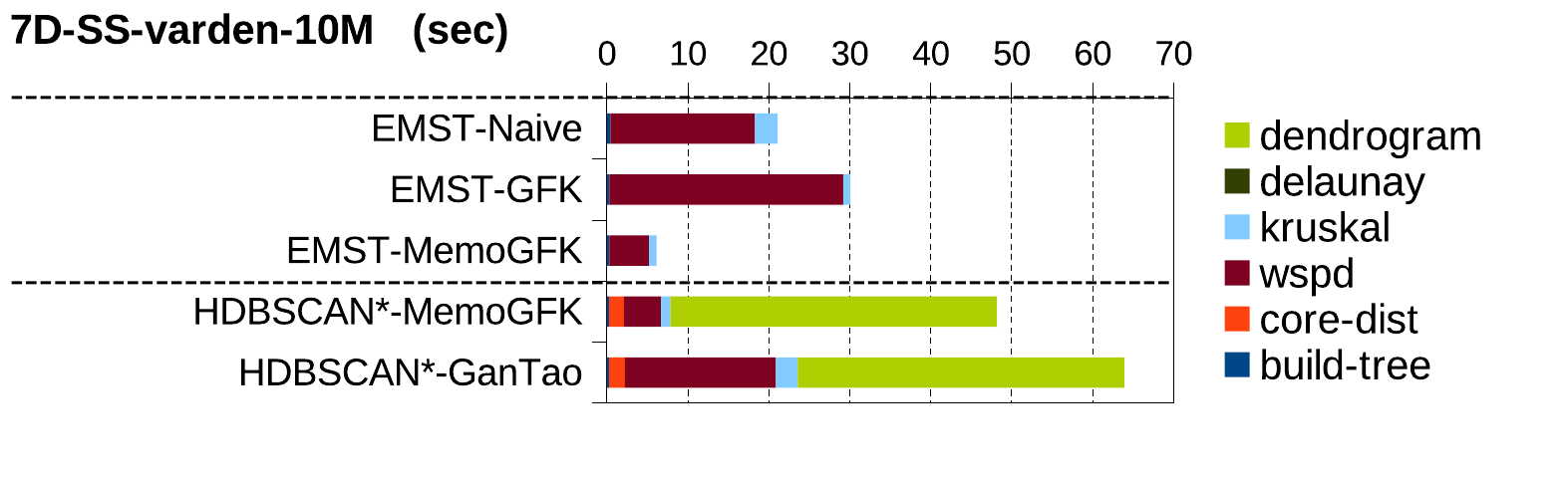}
    \end{subfigure} %
    \caption{
      Decomposition of running times for constructing the EMST and \hdbscan
      MST on various data sets  using all 48 cores with hyper-threading.
      $\minpts=10$ for \hdbscan.
      In the legend,
      "dendrogram" refers to computing the \newddg{};
      "delaunay" refers to computing the Delaunay triangulation;
      "kruskal" refers to Kruskal's MST algorithm;
      "wspd" refers to computing the WSPD decomposition, or the sum of WSPD tree traversal times across rounds;
      "core-dist" refers to computing core distances of all points; and
      "build-tree" refers to building a \kdt{} on all points.
    }
    \label{plot:bar_plots}
  \end{center}
\end{figure*}

%% file: results_bar.tex
\begin{figure}
\vspace{-3pt}
  \begin{center}
    \begin{subfigure}[b]{3in}
      \includegraphics[trim=0 15 0 5, clip,width=0.9\textwidth]{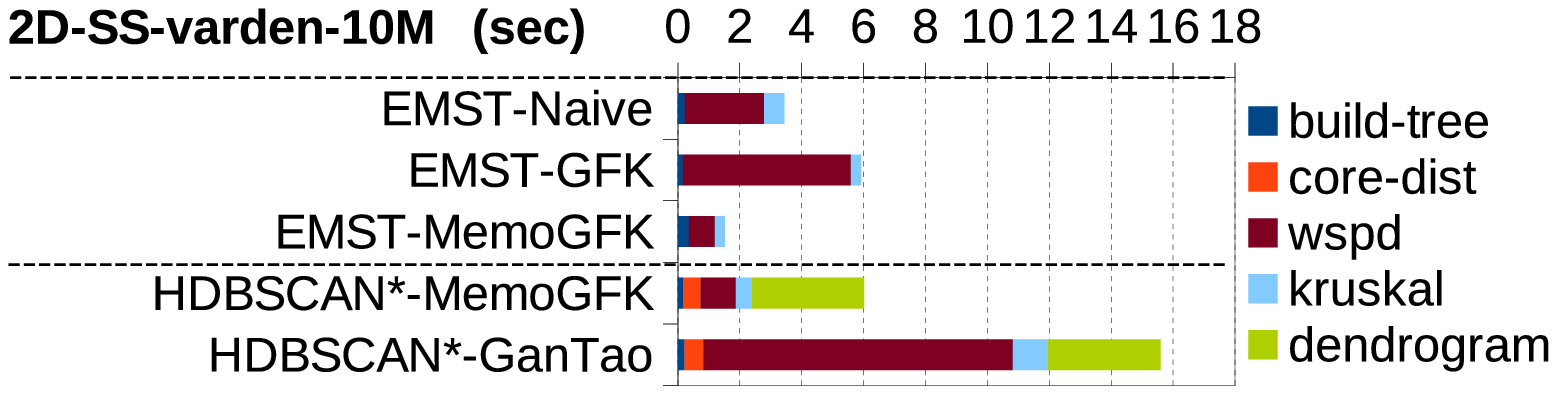}
    \end{subfigure} %
    \begin{subfigure}[b]{3in}
      \includegraphics[trim=0 15 88 5, clip,width=0.72\textwidth]{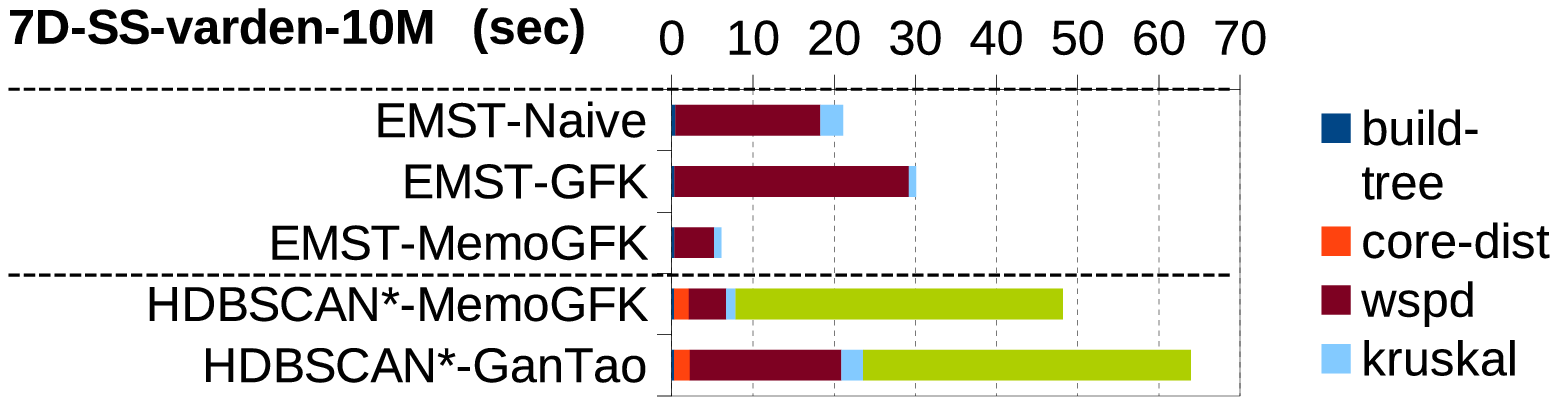}
    \end{subfigure} %
    \caption{
      Decomposition of running times for  EMST and \hdbscan ($\minpts=10$)
      on two data sets using 48 cores with hyper-threading.
      "dendrogram" refers to computing the \newddg{};
      ``kruskal'' refers to Kruskal's MST algorithm;
      "wspd" refers to computing the WSPD decomposition, or the sum of WSPD tree traversal times across rounds;
      "core-dist" refers to computing core distances of all points; and
      "build-tree" refers to building a \kdt{} on all points.
    }
    \label{plot:bar_plots}
  \end{center}
\end{figure}

%% file: speedup_table.tex
{\begin{table}[!t]
\setlength{\tabcolsep}{3pt}
\footnotesize
\begin{tabular}{|l|c|c|c|c|}
\hline
\multicolumn{1}{|c|}{}       & \multicolumn{2}{c|}{\begin{tabular}[c]{@{}c@{}}Speedup over Best Sequential\end{tabular}} & \multicolumn{2}{c|}{\begin{tabular}[c]{@{}c@{}}Self-relative Speedup\end{tabular}} \\ \hline
\multicolumn{1}{|c|}{Method} & Range                                            & Average                                  & Range                                         & Average                              \\ \hline
EMST-Naive                   & $3.51$-$10.69$x                                  & $6.90$x                                  & $16.79$-$33.47$x                              & $24.15$x                             \\ \hline
EMST-GFK                     & $1.52$-$7.01$x                                   & $3.60$x                                  & $8.11$-$11.51$x                               & $9.08$x                              \\ \hline
EMST-MemoGFK                 & $14.61$-$55.89$x                                 & $31.31$x                                 & $14.61$-$55.89$x                              & $31.31$x                             \\ \hline
\iffullversion
Delaunay                     & $14.12$-$16.64$x                                 & $15.38$x                                 & $33.11$-$34.54$x                              & $33.82$x                             \\ \hline
\fi
\hdbscan-MemoGFK                 & $11.13$-$46.69$x                                  & $26.29$x                                 & $11.13$-$46.69$x                              & $26.29$x                             \\ \hline
\hdbscan-GanTao                  & $4.29$-$35.14$x                                  & $13.76$x                                 & $8.23$-$40.32$x                               & $20.97$x                             \\ \hline
\end{tabular}
\caption{Speedup over the best sequential algorithm as well as the self-relative speedup on 48 cores.
}\label{tab:speedup}
\end{table}
}

%% file: results_dendrobar.tex
\begin{figure}[!t]
  \begin{center}
\vspace{-8pt}
\includegraphics[width=0.4\textwidth]{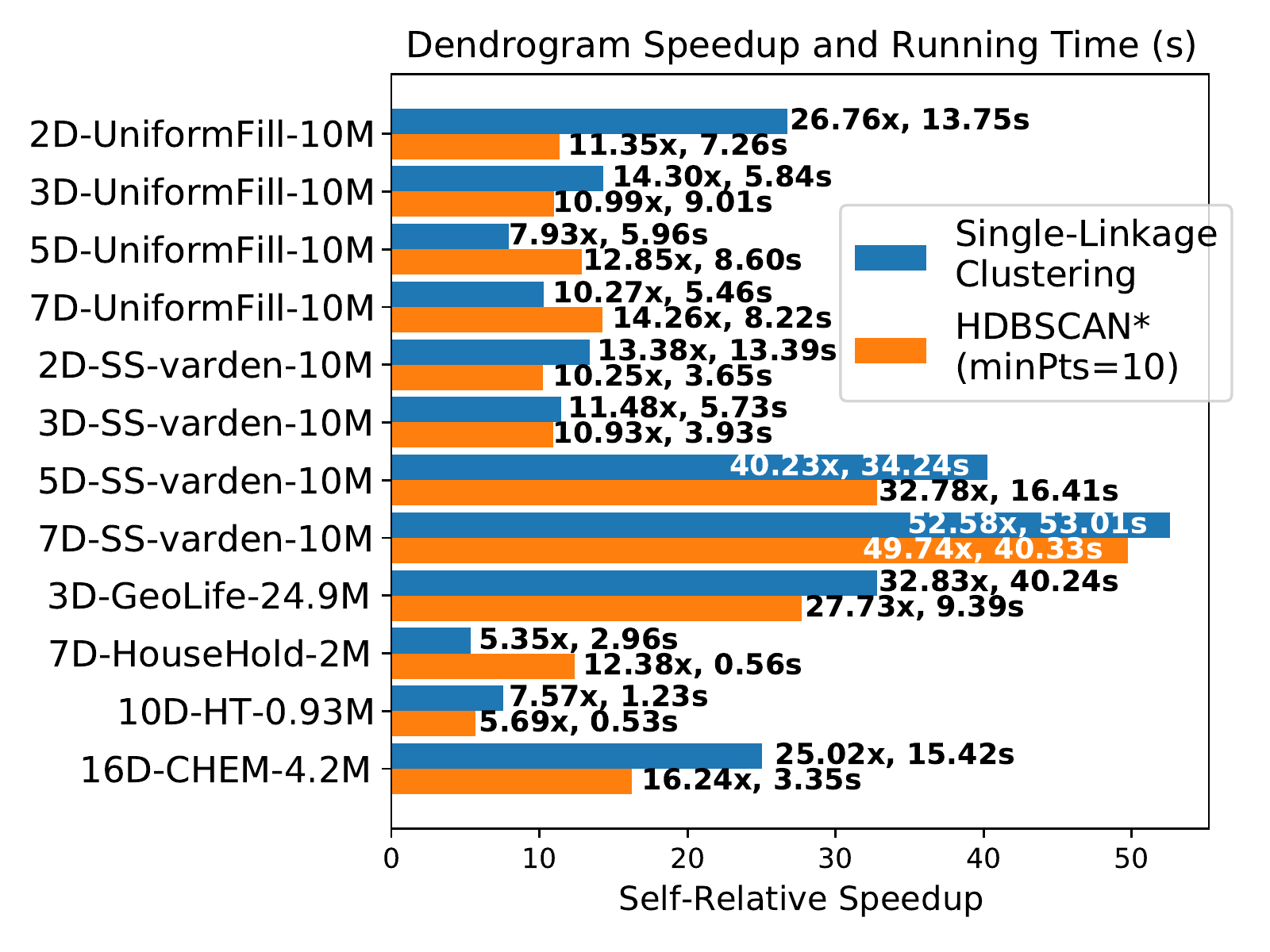}
\vspace{-5pt}
\caption{Self-relative speedups and times for \newddg{}
  computation for single-linkage clustering and \hdbscan ($\minpts=10$). The
  $x$-axis indicates the self-relative speedup on 48 cores with
  hyper-threading. The speedup and time is shown at the
  end of each bar.}
\label{plot:dendro-bar}
\end{center}
\end{figure}

%% file: conclusion.tex
\section{Conclusion}\label{sec:conclusion}
We presented practical and theoretically-efficient parallel algorithms
of EMST and \hdbscan. We also presented a work-efficient
parallel algorithm for computing an ordered dendrogram and
reachability plot. Finally, we showed that our optimized implementations achieve
good scalability and outperform state-of-the-art implementations for
EMST and \hdbscan.

%% file: plane_algorithms.tex
\section{Parallel EMST and \hdbscan in 2D}\label{sec:plane_algo}
\subsection{Parallel EMST in 2D}\label{sec:emst_2d}

The \defn{Delaunay triangulation} on a set of points in 2D contains
triangles among every triple of points $p_1$, $p_2$, and $p_3$ such
that there are no other points inside the circumcircle defined by
$p_1$, $p_2$, and $p_3$~\cite{BCKO}.

In two dimensions, Shamos and Hoey~\cite{Shamos1975Closest} show that the EMST
can be computed by computing an MST on the Delaunay triangulation of the points.
Parallel Delauny triangulation can be computed in
$O(n \log n)$ work and $O(\log n)$ depth~\cite{Reif1992}, and has $O(n)$ edges, and so the MST computation requires the same work and depth.
We provide an implemenation of this algorithm using the parallel
Delaunay triangulation and parallel implementation of Kruskal's
algorithm from the Problem Based Benchmark Suite~\cite{BFGS12}.

\input{hdbscan_2d}

%% file: hdbscan_2d.tex
\subsection{Parallel \hdbscan in 2D} \label{sec:2dhdbscan}

The \defn{ordinary Voronoi
  diagram} is a planar subdivision of a space where points in each
cell share the same nearest neighbor.  The \defn{$k$-order Voronoi
  Diagram} is a generalization of the ordinary Voronoi diagram, where
points in each cell share the same $k$-nearest
neighbors~\cite{aurenhammer1991voronoi}.  A \defn{$k$-order edge} is a
closely related concept, and defined to be an edge where there exists
a circle through the two edge endpoints, such that there are at most
$k$ points inside the circle~\cite{gudmundsson2002higher}. 

De Berg et al.~\cite{BergGR17} show that in two dimensions, the MST on
the mutual reachability graph can be computed in $O(n\log n)$ work.
Their algorithm computes an MST on a graph containing the $k$-order
edges, where $k = \minpts-3$, and where the edges are weighted by the
mutual reachability distances between the two endpoints. They prove
that the MST returned is an MST on the mutual reachability graph.  In
this section, we extend their result to the parallel setting.

To parallelize the algorithm, we need to compute the $k$-order edges
of the points in parallel.  This can be done by first computing the
$(k+1)$-order Voronoi diagram, and then converting the edges in the
Voronoi diagram to $k$-order edges, as shown by Gudmundsson et al.~\cite{gudmundsson2002higher}.
Specifically, we convert each Voronoi edge into a $k$-order edge by
connecting the two points that induce the two cells sharing the Voronoi edge.

Meyerhenke~\cite{meyerhenke2005constructing} shows that the family of
the order-$j$  Voronoi diagrams for all $1\leq j\leq k$ can be computed
in $O(k^2 n\log n)$ work and  $O(k\log^2 n)$ depth.
The algorithm
works by first computing the ordinary Voronoi diagram on
the input points. Then for each Voronoi cell, it computes the ordinary Voronoi
diagram again on the points that induce the neighboring cells. This
ordinary Voronoi diagram divides the Voronoi cell into multiple
subcells,
each of which corresponds to a cell in the Voronoi diagram of one higher order.
This process is repeated until obtaining the order-$k$ Voronoi diagram.
Lee~\cite{lee1982k} proves that the number of $k$-order edges is $O(nk)$, and so we can run parallel MST on these edges in $O(nk\log n)$ work and $O(\log n)$ depth.
This gives us the following theorem.

\begin{theorem}
  Given a set of $n$ points in two dimensions, we can compute
the MST on the mutual reachability graph in $O(\minpts^2 \cdot n\log n)$ work and
  $O(\minpts \cdot \log^2 n)$ depth.
\end{theorem}

For computing the ordinary Voronoi diagrams on each step of
Meyerhenke's algorithm, we use the parallel Delaunay triangulation
implementation from the Problem Based Benchmark Suite~\cite{BFGS12} and take the dual of the resulting triangulation. 
However, we found it to be significantly slower than our other methods
due to high work of the Voronoi diagram computations.

%% file: subquadratic-appendix.tex
\section{Subquadratic-work Parallel EMST}\label{sec:subquadratic-emst}

Callahan and Kosaraju's algorithm~\cite{CallahanK93} first constructs a fair-split tree $T$
and its associated WSPD in $O(n\log n)$ work and $O(\log n)$ depth~\cite{callahan1993optimal},
which is improved from a previous version with $O(n \log n)$ work and $O(\log^2 n)$ depth~\cite{CallahanK95}.
Then, the algorithm runs Boruvka's steps for $\lceil \log_2 n \rceil$ rounds.
In particular, in each round,
the algorithm finds the lightest outgoing edges only for the components with size at most
$2^{i+1}$, and merges the components connected by these edges.
To do so, the algorithm constructs for every component a set of candidate
points that contains the nearest point outside the component.
The algorithm searches for the candidates top-down on $T$, and maintains for each node
in the tree,
a list of all the component s that can have candidates in the subtree of that node.
They ensure the size of each list is $O(1)$ using the WSPD in a manner
identical to the all-nearest-neighbors algorithm of~\cite{CallahanK95}.
In this process, they push the lists down to the leaves of $T$,
so that the candidates corresponding to a component will be the leaves
that contain that component in their lists.

Let $P_j$ be the set of candidates for the $j$'th component.
$P_j$ is split into $\lceil |P_j|/2^{i+1} \rceil$ subsets of
size at most $2^{i+1}$ each, and the $\bccp$ is found between each subset
and the $j$'th component.
At round $i$, there are $n/2^i$ components, and
the $\bccp$ routine is invoked $\sum_{j=1}^{n/2^i} \lceil |P_j|/2^{i+1} \rceil = O(n/2^i)$ times, each with size at most $2^{i+1}$. 
Therefore, the work for $\bccp$ on each round is
$O((n/2^i) T_d(2^{i+1},2^{i+1}))$.
Since $T_d$ is at least linear, this dominates the work for each phase.
The total work for $\bccp$ computations is $O(T_d(n,n)\log n)$.

In our parallel algorithm, on each round, we perform both the
candidate listing step and $\bccp$ computations in parallel.  Listing
candidates for all components can be computed in parallel given a WSPD.
In particular, this uses the top-down computation used for the all-nearest-neighbor
search, parallelized using rake and compress operations~\cite{CallahanK95}, and takes logarithmic depth.
Wang et al.~\cite{wang2019dbscan} show that $\bccp$ can be computed in
parallel in $O(n^{2-(2/(\lceil d/2\rceil + 1))+\delta})$ expected work and
$O(\log^2 n\log^* n)$ depth \whp{}. Both the work and depth at each round is
therefore dominated by computing the $\bccp$s. With $O(\log n)$ rounds,
this results in $O(T_d(n,n) \log n)$ expected work and $O(\log^3 n\log^* n)$ depth \whp{},
where $T_d(n,n) = O(n^{2-(2/(\lceil d/2\rceil + 1))+\delta})$.  This
is also the work and depth of the overall EMST algorithm, as WSPD
construction only contributes lower-order terms to the complexity.

\begin{theorem}
We can compute the EMST on a set of $n$ points in $d$ dimensions in $O(n^{2-(2/(\lceil d/2\rceil + 1))+\delta}\log n)$ expected work and polylogarithmic depth \whp{}.
\end{theorem}

%% file: approx_optics.tex
\section{Parallel Approximate OPTICS}\label{section:approx}


\myparagraph{Parallel Algorithm}
Gan and Tao~\cite{Gan2018} propose a sequential algorithm to solve
the approximate OPTICS problem, defined in \textsc{Lemma 4.2} of their paper~\cite{Gan2018} (this also gives an approximation to \hdbscan).
The algorithm takes in an additional parameter $\rho \geq 0$, which is related to the approximation factor.
The algorithm makes use of the WSPD and uses $O(n\cdot \minpts^2)$ space, with the separation constant $s=\sqrt{8/\rho}$. 
They construct a base graph by adding $O(\minpts^2)$ edges between each
well-separated pair,
and then
compute an MST on the resulting graph. Their algorithm takes $O(n\log n)$
work (where the dominant cost is computing the WSPD). We observe that
their algorithm can be parallelized by plugging in parallel WSPD and
MST routines, resulting in an $O(n\log n)$ work
and $O(\log^2 n)$ depth algorithm.

In parallel for all well-separated pairs, we compute an approximation to
the $\bccp$ for each pair. 
The algorithm uses the rake-and-compress algorithm of Callahan and Kosaraju~\cite{CallahanK95}
to obtain a set of coordinates used for approximation for every subset of the decomposition tree,
taking $O(n)$ work and $O(\log n)$ depth.
Then, for every pair in parallel, our algorithm computes the approximate $\bccp$ in constant work and depth,
similar to the sequential algorithm described in~\cite{CallahanK93}. Overall, this step takes linear work and $O(\log n)$ depth~\cite{CallahanK93,callahan1993optimal}. 

Similarly to Gan and Tao, we call the pair of points in the approximate $\bccp$ of each pair the \defn{representative points}.
For each well-separated pair $(A,B)$, there are four cases for generating edges between $A$ and $B$:
(a) if $|A|< \minpts$ and $|B|< \minpts$, then all pairs of points
between $A$ and $B$ are connected; (b) if $|A|\geq \minpts$ and $|B|< \minpts$, then the representative point
of $A$ is connected to all points in $B$; (c)
if $|A|< \minpts$ and $|B|\geq \minpts$, then the representative point
of $B$ is connected to all points in $A$; and (d)
if $|A|\geq \minpts$ and $|B|\geq \minpts$,
then only the  representative points are connected.
The weight of the edges are
\begin{align*}
  w(u,v) = \max\{\CD(u),\CD(v),\frac{d(u,v)}{1+\rho}\}
\end{align*}
given representative points $u$ and $v$.
In our implementation, we simplify the approximate \bccp{}
by simply picking a random pair of points from each well-separated pair, and also use the parallel MST algorithm introduced in Sections~\ref{sec:gfk} and~\ref{sec:impl:memogfk}, which computes the approximate \bccp{} on the fly for well-separated pairs that are not yet connected.
This will take $O(\minpts^2\cdot n)$ work
due to $O(\minpts^2)$ edges produced for each pair, and $O(\log^2 n)$ depth.
This gives us \cref{thm:approx}.

\begin{theorem}\label{thm:approx}
Given a set of $n$ points, we can compute the MST required for
approximate OPTICS in $O(n\log n)$ work, $O(\log^2 n)$ depth, and
$O(n\cdot \minpts^2)$ space.
\end{theorem}

\input{results_hdbscan_approx}

\myparagraph{Experimental Results}
We study the performance of our parallel implementation for
the approximate OPTICS problem, which we call \defn{OPTICS-GanTaoApprox}.
It uses the
MemoGFK optimization described in Section~\ref{sec:impl:memogfk}.
We found that when run with a reasonable parameter of $\rho$ 
that leads to good clusters,
OPTICS-GanTaoApprox is usually slower than our exact version of the
algorithm (\hdbscan-GanTao, described in Section~\ref{sec:hdbscan_wspd:existing}).
The primary reason is that a
reasonable $\rho$ value requires a high separation constant in the
WSPD, which produces a very large number of well-separated pairs,
leading to poor performance.  In contrast, in the exact algorithm, a
small separation constant ($s=2$) is sufficient for correctness.
Figure~\ref{plot:hdbscan_approx} shows the speedups on two data sets
for OPTICS-GanTaoApprox with $\rho = 0.125$ (corresponding to a
separation constant of $8$) compared with other methods.
Across all of the data sets, we found 
OPTICS-GanTaoApprox to be slower than 
\hdbscan-GanTao by a factor of 1.00--1.96x, and slower than \hdbscan-MemoGFK by a factor of 1.72--7.48x.

%% file: results_hdbscan_approx.tex
\begin{figure}[!t]
\begin{center}
\includegraphics[width=0.49\textwidth]{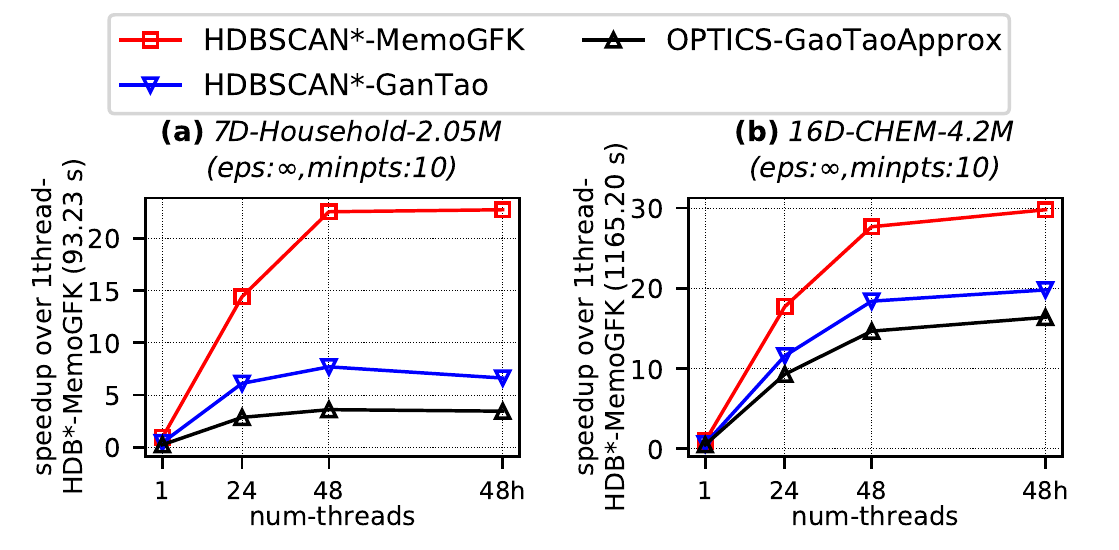}
\caption{Speedup of \hdbscan MST implementations over the \emph{best} serial baselines vs.
  thread count.
The best serial baseline and its running time for each data set
  is shown on the $y$-axis label. ``48h''
  on the $x$-axis refers to 48 cores with hyper-threading.
}
\label{plot:hdbscan_approx}
\end{center}
\end{figure}

%% file: minpts_ex.tex
\section{Relationship between EMST and \hdbscan MST}\label{sec:minpts_ex}

We now show that for $\minpts \leq 3$, the EMST is always an MST of the \hdbscan base graph by having the same set of edges,
but for higher values of \minpts it is possible that this is not the case.
For example, \cref{fig:minpts_ex} gives an example
where EMST is not an MST of the \hdbscan base graph when $\minpts = 4$.

\begin{theorem}
An EMST is always an MST for the \hdbscan mutual reachability graph when $\minpts \leq 3$. 
\end{theorem}
\begin{proof}
For $\minpts \leq 2$, all edges in the \hdbscan mutual reachability graph have edge weights defined by
Euclidean distances, and so the edge weights are identical. We now discuss the case when $\minpts = 3$.

Let $T$ be an EMST, and $T' \neq T$ be an MST in $G_{MR}$.  We show that we can convert $T'$ to $T$ without changing the total weight. 
Consider any edge $(u,v) \in T$, but not in $T'$. If we add $(u,v)$ to $T'$, then we get a cycle $C$.

First, we show that $(u,v)$ cannot be the unique heaviest edge in $C$  under $d_m$. 
Recall that
$\CD(p)$ is the core distance of a point $p$ and $d_m(p,q) =
\max\{\CD(p),\CD(q), d(p,q)\}$. 
Assume by contradiction that $(u,v)$ is the unique heaviest edge in $C$ under $d_m$.

If $d_m(u,v) = d(u,v)$, then $(u,v)$ is also the unique heaviest edge 
in $C$ in the Euclidean complete graph, and so it cannot be in $T$, which is the EMST. This is a contradiction.
 
Now we consider the case where $d_m(u,v) > d(u,v)$. Without loss of generality, suppose that $d_m(u,v) = \CD(u)$. 
Then $v$ must be $u$'s unique nearest neighbor; otherwise, $d_m(u,v)=\CD(u)=d(u,v)$ because we have $\minpts = 3$. However, then all other points have larger distance to $u$ than $d(u,v)$, and $u$ must have an edge to one of these other points in the cycle $C$. Thus, $(u,v)$ cannot be the unique heaviest edge in $C$. This is a contradiction. 

Now, given that $(u,v)$ is not the unique heaviest edge in $C$, we can replace one of the heaviest edges $e$ that is in $C$, but not in $T$, with $(u,v)$, and obtain another MST in $G_{MR}$ with the same weight. 

Below we show that there is always such an edge $e$ in $C$. 
We first argue that there must be some heaviest edge in $C$ that has its Euclidean distance as its weight in $G_{MR}$. 
Consider a heaviest edge $(a,b)$ in $C$, and without loss of generality, suppose that $d_m(a,b) = \CD(a)$.
If $(a,b)$ does not have its Euclidean distance as its edge weight, then $b$ must be $a$'s unique nearest neighbor. 
Besides $(a,b)$, $a$ must be incident to another edge in $C$, which we denote as $(a,c)$. $d_m(a,c)$ must equal $d_m(a,b)$: we have
$d_m(a,c) \geq \CD(a) = d_m(a,b)$ because $b$ is $a$'s unique nearest neighbor, but we also have $d_m(a,c) \leq d_m(a,b)$ because $(a,b)$ is a heaviest edge in $C$. Therefore, $d_m(a,c)=d_m(a,b)$, and $(a,c)$ is one of the heaviest edges in $C$ under $d_m$. Furthermore, $d_m(a,c)=d(a,c)$ because $d(a,c) \leq d_m(a,c)$ by definition and $d(a,c) \geq \CD(a) = d_m(a,b) = d_m(a,c)$ because $\minpts = 3$ and $b \neq c$ is $a$'s unique nearest neighbor. 
Thus, we have shown that $(a,c)$ is a heaviest edge in $C$ that has its Euclidean distance as its weight in $G_{MR}$.

All heaviest edges that have the Euclidean distance as their weight must also be  the heaviest edges in $C$ in the Euclidean complete graph, and thus they cannot all be in the EMST $T$. Therefore, there must exist some heaviest edge $e \in C$ that is in $T'$ but not in $T$.
We can always find such an edge in $T'$ and swap it with the edge $(u,v)$ in $T$ to make $T'$ share more edges with $T$, without changing the total weight of $T'$ in $G_{MR}$, as both edges are heaviest edges in $C$ under $d_m$. 
We can repeat this process until we obtain $T$. Therefore, $T$ is also an MST in $G_{MR}$.
\end{proof}

\begin{figure}
  \includegraphics[width = 0.5\columnwidth]{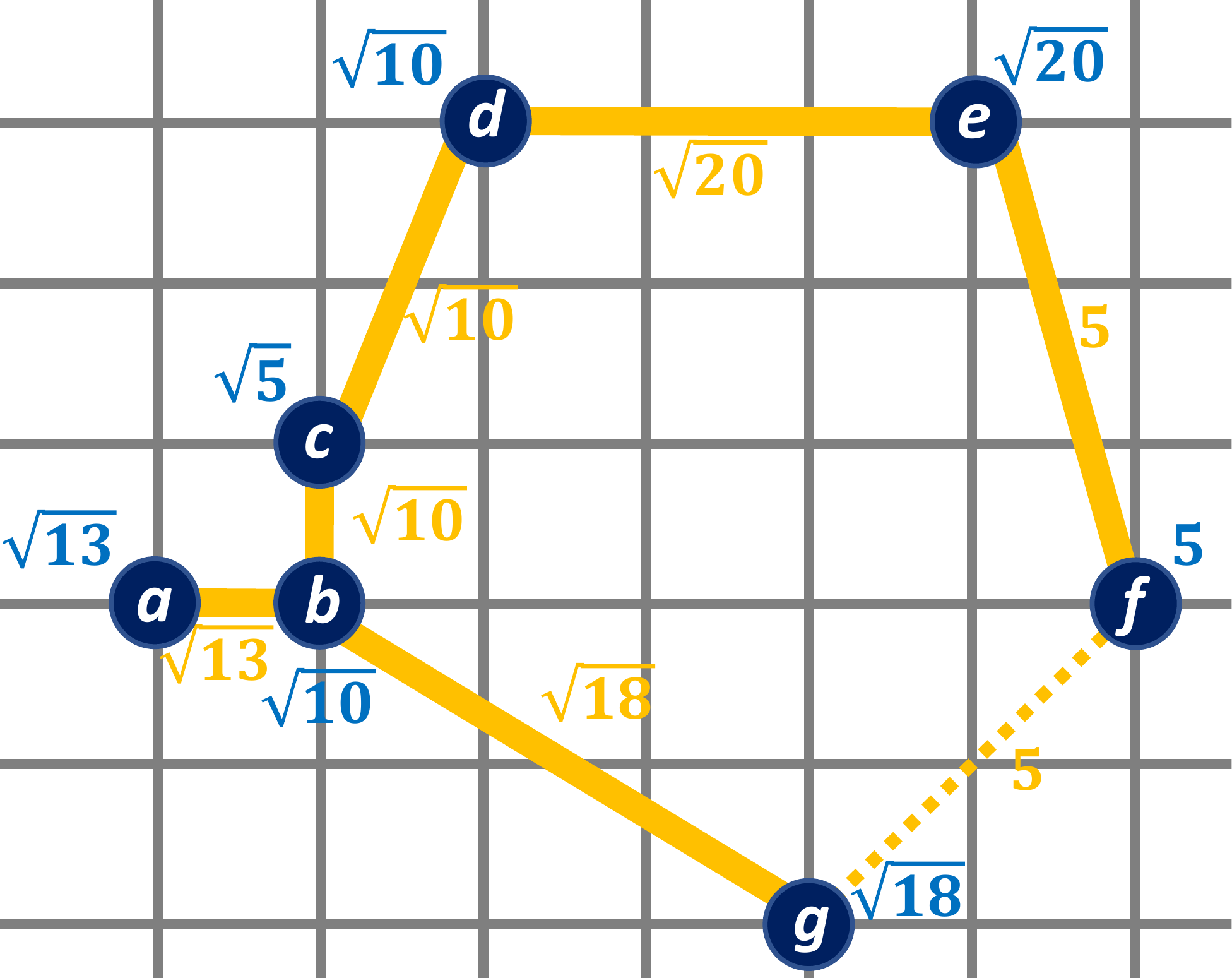}
  \caption{An example where the EMST is not an MST of \hdbscan base graph (call it MST$^*$), when $\minpts = 4$.
The blue values are core distances of the points.
The yellow values are weights of edges according to their mutual reachability distances.
The solid edges form MST$^*$. Both edge $(f,g)$ and $(e,f)$ are in the EMST, but cannot both be in MST$^*$ because they are the heaviest edges in the $g\mhyphen b\mhyphen  c\mhyphen  d\mhyphen e\mhyphen f\mhyphen g$ cycle.
}
    \label{fig:minpts_ex}
\end{figure}

%% file: tables-appendix.tex
\section{Additional Data from Experiments}
\cref{table:emst_mlpack} shows the running times for \texttt{mlpack}. \cref{table:emst} shows the running times of our implementations for EMST. \cref{table:hdbscan} shows the running times of our implementations for \hdbscan.

%% file: table_timing_mlpack.tex
\begin{table*}[!t]
\small
\begin{tabular}{|c|c|}
   \hline
& \texttt{mlpack} (1 thread)    \\\hline
2D-UniformFill-10M & 90.09 \\\hline
3D-UniformFill-10M & 211.04 \\\hline
5D-UniformFill-10M & 964.13 \\\hline
7D-UniformFill-10M & 4777.29 \\\hline
2D-SS-varden-10M & 84.79 \\\hline
3D-SS-varden-10M & 139.18 \\\hline
5D-SS-varden-10M & 184.08 \\\hline
7D-SS-varden-10M & 233.28 \\\hline
3D-GeoLife-10M & 211.37 \\\hline
7D-Household-2.05M & 59.15 \\\hline
10D-HT-0.93M & 14.85 \\\hline
16D-CHEM-4.2M & 732.6 \\\hline
\end{tabular}
\caption{Table of running times in seconds for the sequential EMST implementation from \texttt{mlpack}.}\label{table:emst_mlpack}
\end{table*}

%% file: table_timing_emst.tex
\begin{table*}[!htbp]
   \small

\begin{tabular}{|c|c|c|c|c|c|c|c|c|}
\hline
\multicolumn{9}{|c|}{EMST}                                                                                                                               \\ \hline
                   & \multicolumn{2}{c|}{EMST-Naive} & \multicolumn{2}{c|}{EMST-GFK} & \multicolumn{2}{c|}{EMST-MemoGFK} & \multicolumn{2}{c|}{Delaunay} \\ \hline
                   & 1 thread      & 48 cores      & 1 thread     & 48 cores     & 1 thread      & 48 cores        & 1 thread     & 48 cores     \\ \hline
2D-UniformFill-10M & $62.51$       & $3.64$          & $57.93$      & $6.11$         & $31.54$       & \textbf{1.20}     & $65.46$      & $1.90$         \\ \hline
3D-UniformFill-10M & $400.57$      & $19.30$         & $218.02$     & $26.07$        & $67.80$       & \textbf{2.24}     & --            & --              \\ \hline
5D-UniformFill-10M & --             & --               & --            & --              & $230.93$      & \textbf{5.03}     & --            & --              \\ \hline
7D-UniformFill-10M & --             & --               & --            & --              & $1585.65$     & \textbf{28.37}    & --            & --              \\ \hline
2D-SS-varden-10M   & $57.84$       & $3.45$          & $60.64$      & $6.90$         & $27.48$       & \textbf{1.60}     & $64.42$      & $1.95$         \\ \hline
3D-SS-varden-10M   & $240.24$      & $12.13$         & $189.52$     & $23.37$        & $48.72$       & \textbf{1.85}     & --            & --              \\ \hline
5D-SS-varden-10M   & $478.40$      & $19.41$         & $278.10$     & $31.19$        & $96.19$       & \textbf{3.04}     & --            & --              \\ \hline
7D-SS-varden-10M   & $626.78$      & $21.10$         & $336.62$     & $29.26$        & $205.15$      & \textbf{6.18}     & --            & --              \\ \hline
3D-GeoLife-10M     & $271.95$      & $10.97$         & $328.76$     & $36.31$        & $117.31$      & \textbf{4.77}     & --            & --              \\ \hline
7D-Household-2.05M & $280.28$      & $8.37$          & $214.08$     & $24.77$        & $37.60$       & \textbf{1.40}     & --            & --              \\ \hline
10D-HT-0.93M       & $19.28$       & $0.64$          & $12.36$      & $1.40$         & $5.17$        & \textbf{0.35}     & --            & --              \\ \hline
16D-CHEM-4.2M      & --             & --               & --            & --              & $821.81$      & \textbf{19.11}    & --            & --              \\ \hline
\end{tabular}

\caption{Table of running times in seconds for EMST. The fastest parallel time for each data set is in bold.
  The tests that do not complete within 3 hours or that run out of memory are shown as ``--".
The data sets with dimensionality greater than 2 are not applicable to Delaunay, and also shown as ``--''.}\label{table:emst}
\end{table*}

%% file: table_timing_hdbscan.tex
\begin{table*}[!htbp]
   \small

\begin{tabular}{|c|c|c|c|c|}
\hline
\multicolumn{5}{|c|}{\hdbscan ($\minpts=10$)}                                                           \\ \hline
                   & \multicolumn{2}{c|}{\hdbscan-MemoGFK} & \multicolumn{2}{c|}{\hdbscan-GanTao} \\ \hline
                   & 1 thread     & 48 cores           & 1 thread            & 48 cores           \\ \hline
2D-UniformFill-10M & $197.55$     & \textbf{10.34}     & $298.03$            & $18.71$            \\ \hline
3D-UniformFill-10M & $321.97$     & \textbf{14.66}     & $517.71$            & $24.04$            \\ \hline
5D-UniformFill-10M & $1217.87$    & \textbf{38.41}     & $2395.68$           & $68.54$            \\ \hline
7D-UniformFill-10M & $7487.95$    & \textbf{289.27}    & --                   & --                  \\ \hline
2D-SS-varden-10M   & $103.73$     & \textbf{6.07}      & $163.37$            & $15.66$            \\ \hline
3D-SS-varden-10M   & $154.31$     & \textbf{7.56}      & $253.06$            & $15.44$            \\ \hline
5D-SS-varden-10M   & $716.81$     & \textbf{22.20}     & $885.92$            & $34.51$            \\ \hline
7D-SS-varden-10M   & $2253.38$    & \textbf{48.26}     & $2585.83$           & $64.13$            \\ \hline
3D-GeoLife-10M     & $687.75$     & \textbf{22.70}     & $1320.15$           & $160.48$           \\ \hline
7D-Household-2.05M & $93.23$      & \textbf{3.54}      & $204.75$            & $13.51$            \\ \hline
10D-HT-0.93M       & $15.74$      & \textbf{1.41}      & $29.75$             & $3.21$             \\ \hline
16D-CHEM-4.2M      & $1165.20$    & \textbf{35.77}     & $1820.61$           & $55.52$            \\ \hline
\end{tabular}

\caption{Table of running times in seconds for \hdbscan with $\minpts=10$. The fastest parallel time for each data set is in bold.
  The tests that do not complete within 3 hours, or that run out of memory are shown as ``--''.
}\label{table:hdbscan}
\end{table*}